\documentclass[11pt]{article}

\usepackage{microtype}%
\usepackage{comment}

\usepackage{amssymb}
\usepackage{amsmath}
\usepackage{amsthm}
\usepackage{fullpage}
\usepackage{color}  %

\title{The Robustness of LWPP and WPP,\\with an Application to Graph
  Reconstruction}
\author{Edith Hemaspaandra\\
Department of Computer Science\\Rochester Institute of Technology\\Rochester, NY 14623, USA
\and
Lane A. Hemaspaandra\thanks{This work was done in part
    while on a sabbatical stay at ETH Z\"{u}rich's Department of
    Computer Science, generously supported
    by that department.}\\
Department of Computer Science\\University of Rochester\\ Rochester, NY
  14627, USA
\and
Holger Spakowski\thanks{This work was done in part while visiting the University of Rochester.}\\
Department of Mathematics and Applied Mathematics\\
  University of Cape Town\\Rondebosch 7701, South Africa
\and
Osamu Watanabe\\
Dept.\ of Mathematical and Computing Sciences\\Tokyo Institute of Technology\\
Tokyo 152-8552, Japan}

\date{November 3, 2017; revised
July 6, 2018}

\newcommand{\ceqp}{\mbox{$\rm C_{=}P$}}
\newcommand{\coceqp}{\mbox{$\rm coC_{=}P$}}
\newcommand{\indexedlwpp}[1]{\mbox{\rm $#1\mhyphen\lwpp $}}
\newcommand{\indexedceqp}[1]{\mbox{\rm $#1\mhyphen\ceqp $}}
\newcommand{\indexedwpp}[1]{\mbox{\rm $#1\mhyphen\wpp $}}
\newcommand{\sigmastar}{\Sigma^*}
\newcommand{\eqp}{\mbox{\rm EQP}}
\newcommand{\pe}{\mbox{\rm P}}
\newcommand{\MOD}{\mbox{\rm MOD}}

\newcommand{\upleqtwo}{\mbox{\rm UP}_{\leq 2}}

\newcommand{\lwpp}{\mbox{\rm LWPP}}
\newcommand{\wpp}{\mbox{\rm WPP}}
\newcommand{\spp}{\mbox{\rm SPP}}
\newcommand{\ld}{\mbox{\rm Legitimate Deck}}
\newcommand{\twolwppplus}{\mbox{\rm Two-$\lwpp^+$}}
\newcommand{\np}{\mbox{\rm NP}}
\newcommand{\pp}{\mbox{\rm PP}}
\newcommand{\sharpp}{\mbox{\rm \#P}}
\newcommand{\gapp}{\mbox{\rm GapP}}
\newcommand{\fp}{\mbox{\rm FP}}
\newcommand{\range}{\mbox{\rm range}}
\newcommand{\indexedpolylwpp}{\mbox{\rm
    ${\mathit{Poly}}\mhyphen\lwpp$}}
\newcommand{\indexedpolyceqp}{\mbox{\rm
    ${\mathit{Poly}}\mhyphen\ceqp$}}
\newcommand{\indexedloglwpp}{\mbox{\rm ${\mathit{Log}}\mhyphen\lwpp$}}
\newcommand{\indexedpolywpp}{\mbox{\rm ${\mathit{Poly}}\mhyphen\wpp$}}
\newcommand{\indexedexplwpp}{\mbox{\rm
    ${\mathit{Exp}}\mhyphen\lwpp$}}
\newcommand{\indexedexpwpp}{\mbox{\rm
    ${\mathit{Exp}}\mhyphen\wpp$}}

\newcommand{\multipliedlwpp}[1]{\mbox{\rm ${#1}\mhyphen\widehat{\lwpp} $}}

\newcommand{\integers}{\ensuremath{{ \mathbb{Z} }}}
\newcommand{\naturalnumbers}{\ensuremath{{ \mathbb{N} }}}
\newcommand{\positivenumbers}{\ensuremath{{ \mathbb{N}^{+} }}}
\newcommand{\bigoh}{{\mathcal{O}}}
\newcommand{\calo}{{\mathcal{O}}}
\mathchardef\mhyphen="2D
  \newtheorem{theorem}{Theorem}[section]

\newtheorem{claim}[theorem]{Claim}
\newtheorem{proposition}[theorem]{Proposition}
\newtheorem{conjecture}[theorem]{Conjecture}
  \newtheorem{corollary}[theorem]{Corollary}

  \newtheorem{lemma}[theorem]{Lemma}

  \newtheorem{definition}[theorem]{Definition}

\newtheorem{closurep}[theorem]{Closure Property}

\newcommand{\confl}{{{\rm conflicting}}}
\newcommand{\gap}{\mbox{\rm gap}}
\newcommand{\accept}{\mbox{\rm acc}}
\newcommand{\rej}{\mbox{\rm rej}}
\newcommand{\mono}{\mbox{\rm mono}}
\newcommand{\sign}{\mbox{\rm sign}}
\newcommand{\degree}{\mbox{\rm deg}}
\newcommand{\charfun}{\chi}

\newcommand{\pcount}{\mbox{\rm PCount}}
\newcommand{\gi}{\mbox{\rm GI}}

\begin{document}
\sloppy
\maketitle

\begin{abstract}
We show that the counting class 
LWPP~\cite{fen-for-kur:j:gap}
remains unchanged even if one allows a polynomial number of 
gap values rather than one.
On the other hand, we show that
it is impossible to improve
this from polynomially many gap values to a superpolynomial number of gap values
by relativizable proof techniques.

The first of these results implies that the Legitimate Deck Problem
(from the study of graph reconstruction) is in LWPP (and thus
low for PP, i.e., 
$\rm PP^{\mbox{\scriptsize{}Legitimate Deck}} = PP$) if the weakened
version of the Reconstruction Conjecture holds in which the number of
nonisomorphic preimages is assumed merely to be polynomially bounded.
This strengthens the 1992 result of K\"{o}bler, Sch\"{o}ning, and
Tor\'{a}n~\cite{koe-sch-tor:j:gi-pplow} 
that the Legitimate Deck Problem 
is in LWPP if the Reconstruction Conjecture holds, 
and provides strengthened evidence that 
the Legitimate Deck Problem is not NP-hard.

We additionally show on the one hand that our main LWPP robustness
result also holds for WPP, and also holds even when one allows both
the rejection- and acceptance- gap-value targets to simultaneously be
polynomial-sized lists; yet on the other hand, we show that for the
$\sharpp$-based analog of LWPP the behavior 
much differs in that, in some relativized worlds, even two
target values already yield a richer class than one value does.
Despite that nonrobustness result
for a $\sharpp$-based class, 
we show that the $\sharpp$-based ``exact counting''
class $\ceqp$ remains unchanged even if one allows a 
polynomial number of target values for the 
number of accepting paths of the machine.

\end{abstract}

\section{Introduction}\label{s:intro}
Nothing is more natural than wanting to better understand an object by
knowing what it can and cannot do.  Whether wondering how fast a
(rental?)~car can go in reverse or wondering if $\rm \np^{NP}$ can
without loss of generality be assumed to ask at most one question per
nondeterministic path (as it indeed can, as is implicit in the
quantifier
characterization~\cite{mey-sto:c:word-exponential,wra:j:complete} 
of
$\rm \np^{NP}$), we both in life and as theoreticians want to find how
robust things are.

We are often particularly happy when a class
proves to be quite robust under definitional perturbations.  Such 
robustness on one hand suggests that perhaps there is something 
broadly natural about the class, and on the other hand such robustness
often makes it easier to put the class to use.

This paper shows that the counting classes LWPP and WPP,
defined in 1994 in the seminal work of Fenner, Fortnow, and
Kurtz~\cite{fen-for-kur:j:gap} on gap-based counting classes, are
quite robust.  Even though their definitions are in terms of having the
gap function (the difference between the number of accepting and
rejecting paths of a machine) hit a single target value, we prove 
(in Section~\ref{s:many}) that
one can allow a 
list 
of up to polynomially many target values without
altering the descriptive richness of the class, i.e., without changing
the class.

We then apply this to the question of whether the Legitimate Deck
Problem is in LWPP\@.

\emph{The Legitimate Deck Problem} (a formal 
definition will be given in 
Section~\ref{s:app}) is the decision problem of 
determining whether, given a multiset
of (unlabeled) graphs, there exists a graph $G$ such that that
multiset is precisely (give or take isomorphisms) the multiset of
one-node-deleted subgraphs of $G$ (a.k.a.~the \emph{deck} of
$G$)~\cite{hem-kra:j:legitimate-deck}.  \emph{The Reconstruction
Conjecture}~\cite{kelly:thesis:trans,ulam:b:mat-prob}---which in the
wake of the resolution of the Four-Color Conjecture was declared by
the editorial board of the Journal of Graph Theory to be the foremost
open problem in graph
theory~\cite{edi:j:graph-reconstruction}---states that every graph
with three or more nodes is uniquely determined (give or take
isomorphisms) by its multiset of one-node-deleted subgraphs.  The
Legitimate Deck Problem was defined in 1978 by
Nash-Williams~\cite{nas:b:sel-top-graph}, in his paper that framed the
algorithmic/complexity issues related to reconstructing graphs---such
as telling whether a given deck is legitimate (i.e., is the deck of
some graph).

Our application of our LWPP robustness result to the question of 
whether the Legitimate Deck Problem is in LWPP is the following.
The strongest previous
evidence of the simplicity of the Legitimate Deck Problem is the 1992
result of K\"{o}bler, Sch\"{o}ning, and
Tor\'{a}n~\cite{koe-sch-tor:j:gi-pplow} that the Legitimate Deck
Problem is in LWPP (and thus is PP-low,
i.e.,
$\rm PP^{\mbox{\scriptsize{}Legitimate Deck}} = PP$)
if the
Reconstruction Conjecture 
(i.e., that each deck
whose elements all have at least two nodes has at most 
\emph{one} preimage, give or take isomorphisms) 
holds.
Using this paper's main robustness result as a tool, 
Section~\ref{s:app} proves that
the Legitimate Deck Problem is in LWPP (and thus is PP-low)
if a weakened
version of the Reconstruction Conjecture holds, namely, that each
deck has at most \emph{a polynomial number} of nonisomorphic 
preimages.

This weakened version is not known to be equivalent to the
Reconstruction Conjecture itself.  And so our result 
for the first time gives
a path to proving that the Legitimate Deck Problem is PP-low that does
not require one to, on the way, resolve the foremost open problem in
graph theory~\cite{edi:j:graph-reconstruction}.

We started this section by noting that it is natural to want to know both
flexibilities and limitations of classes.  Our main result is about
flexibility: going from one target gap to instead a polynomial number.
But are we leaving money on the table?  Could we extend our 
result to slightly superpolynomial numbers of target gaps, or even
to exponential numbers of target gaps?  In Section~\ref{s:opt-body-version-TR} 
we note that if the robustness of LWPP 
were to 
hold
up to exponentially many
target gaps, then 
$\np$
would be in LWPP and so
would be PP-low (i.e., $\rm PP^{NP} = NP$); 
yet 
NP is widely suspected not to be PP-low.  We also,
by encoding nondeterministic
oracle Turing machines by low-degree multivariate polynomials
so as to capture the gap functions of those machines,
show 
an oracle relative to which robustness fails for all 
superpolynomial numbers of target gaps; thus, no extension beyond 
this paper's polynomial-number-of-target-gaps robustness result 
for LWPP can be proven by a relativizable proof.  
And for the $\sharpp$-based analogue of LWPP, in Section~\ref{s:body-tr-lwppplus}
we show that 
even allowing two target values yields, in some relativized worlds, a richer class
than one target value.
For the important 
$\sharpp$-based counting class 
$\ceqp$, however, 
we prove in Section~\ref{s:ceqp} 
that the class does not change even if 
one allows not just one but instead 
a polynomial number of 
target values for the number of accepting paths of the underlying
$\sharpp$ function.
We also (in the final part of
Section~\ref{s:many}) 
extend our 
main result to 
show that the simultaneous expansion to polynomial-sized 
lists of both the acceptance \emph{and rejection} target-gap 
lists still, for LWPP and WPP, yields the classes LWPP and $\wpp$.

To summarize: 
In this paper, we prove that LWPP and WPP are robust
enough that they remain unchanged when their single target gap is
allowed to be expanded to a polynomial-sized 
list; we apply this new
robustness of LWPP to show that the PP-lowness of the Legitimate Deck
Problem follows from a weaker hypothesis than was previously known;
we show that our polynomial robustness of LWPP is optimal with
respect to relativizable proofs; and we prove a number of related results 
on limitations and extensions.

\section{Preliminaries}
We first present the definitions of many of the counting classes 
that we will
be speaking of, taking the definitions from the seminal
paper of Fenner, Fortnow, and Kurtz~\cite{fen-for-kur:j:gap}.

\begin{definition}[\cite{fen-for-kur:j:gap}]
\begin{enumerate}
\item For each nondeterministic polynomial-time Turing machine $N$, 
the function $\accept_N: \sigmastar \rightarrow \naturalnumbers$ is defined
such that for every $x\in \sigmastar$, $\accept_N(x)$ equals the number
of accepting computation paths of $N$ on input $x$.
\item For each nondeterministic polynomial-time Turing machine $N$, 
the function $\rej_N: \sigmastar \rightarrow \naturalnumbers$ is defined
such that for every $x\in \sigmastar$, $\rej_N(x)$ equals the number
of rejecting computation paths of $N$ on input $x$.
\item For each nondeterministic polynomial-time Turing machine $N$,
the function $\gap_N: \sigmastar \rightarrow \integers$ is defined
such that for every $x\in \sigmastar$, 
\[ 
   \gap_N(x) = \accept_N(x) - \rej_N(x).
\]
\end{enumerate}
\end{definition}

\begin{definition}[\cite{fen-for-kur:j:gap}]
\[
 \gapp = \{ \gap_N ~|~ N \text{ is a polynomial-time nondeterministic
    Turing machine} \}. 
\]
\end{definition}

\begin{definition}[\cite{hem-ogi:j:closure,fen-for-kur:j:gap}]
$\spp$ is the class of all sets $A$ such that there exists a $\gapp$
  function $g$ such that for all $x\in\sigmastar$,
\begin{align*}
  x \in A & \Longrightarrow g(x) = 1\\
  x \notin A & \Longrightarrow g(x) = 0.
\end{align*}
\end{definition} 

The following class, $\wpp$, is potentially larger than
$\spp$.  Instead of the ``target value''~1 for the case $x \in A$, we
allow a target value $f(x)$, where $f$ may be any polynomial-time
computable function whose image does not contain 0.  $\fp$ denotes
the class of polynomial-time computable functions.

\begin{definition}[\cite{fen-for-kur:j:gap}]
$\wpp$ is the class of all sets $A$ such that there exists a $\gapp$
  function $g$ and a function $f \in \fp$ that maps from $\sigmastar$ to
  $\integers - \{ 0\}$  such that for all $x\in\sigmastar$,
\begin{align*}
  x \in A & \Longrightarrow g(x) = f(x)\\
  x \notin A & \Longrightarrow g(x) = 0.
\end{align*}
\end{definition} 

The class $\lwpp$ is the same as $\wpp$ except that the ``target
function'' $f$ may depend on only the \emph{length} of the input.

\begin{definition}[\cite{fen-for-kur:j:gap}] \label{def:lwpp}
$\lwpp$ is the class of all sets $A$ such that there exists a $\gapp$
  function $g$ and a function $f \in \fp$ that maps from $0^*$ to
  $\integers - \{ 0\}$  such that for all $x\in\sigmastar$,
\begin{align*}
  x \in A & \Longrightarrow g(x) = f(0^{|x|})\\
  x \notin A & \Longrightarrow g(x) = 0.
\end{align*}
\end{definition} 
Here, as usual, for any $x \in\sigmastar$, $|x|$ denotes the
length of $x$, and for any $n \in\naturalnumbers$, $0^n$ is the string
consisting of exactly $n$ zeroes.

We now generalize the definition of $\lwpp$ to the case of having the target of 
the $\gapp$ function be, for members of the set, not a single value but 
a collection of values.

One might expect us to formalize this 
by simply having the polynomial-time computable 
``what is the target'' function output 
a \emph{list} of the 
nonzero-integer targets.
That would work fine 
and be equivalent to what we are about to do, as long as we are 
dealing with lists having at most a polynomial number of elements.  
However, to be able to speak of even longer lists---as will be 
important in our negative results offsetting our main result---we 
use an indexing 
approach, as follows.

\begin{definition}  \label{def:lwpp-indexed}
Let $r$ be any function mapping from $\naturalnumbers$ to
$\naturalnumbers$.
Then the class
$\indexedlwpp{r}$ is the class of all sets $A$ such that there exists a $\gapp$
  function $g$ and a function $f \in \fp$ that maps
 to
  $\integers - \{ 0\}$  such that for each $x\in\sigmastar$,
\begin{align*}
  x \in A & \Longrightarrow \text{ there exists } i \in \{ 1, 2,
  \ldots , r(|x|)\} \text{ such that } g(x) = f(\langle 0^{|x|}, i \rangle)
\\
  x \notin A & \Longrightarrow g(x) = 0.
\end{align*}
\end{definition}

The following class, $\indexedpolylwpp$,
will be central to this paper:
Our main 
result is that this class in fact equals $\lwpp$.  The ``$+\,c$''
in Definition~\ref{def:indexedpolylwpp}
may seem strange at first.  But without it we would have 
a boundary-case pathology at $n=0$, namely, 
the class could not contain any set that 
contains the empty string.\footnote{The ``$+\,c$''
in Definition~\ref{def:indexedpolylwpp}
also, on its surface, would seem to make a difference
at length 1, by allowing lists of size greater
than one; however, one could work around that issue. 
In contrast,
the exclusion of the empty string 
would not be avoidable if our class of polynomials were to be
a class---such as $n^c$---such that all of its members evaluate to $0$ 
at $n=0$.  In any case, our use of $n^c + c$ avoids any special 
worries at lengths $0$ and $1$. And since for every polynomial $p$ there 
is a $c$ such that $(\forall n \in\naturalnumbers)[ p(n) \leq n^c + c]$,
using 
polynomials just of the form $n^c+c$ is in fact not a restriction.}

\begin{definition} \label{def:indexedpolylwpp}
\[
   \indexedpolylwpp = \bigcup_{c \in
     \positivenumbers} \indexedlwpp{(n^c+c)}.
\]
\end{definition}

It is easy to see that 
$\indexedlwpp{1} = \lwpp$, and that, of course,
more flexibility as to targets never removes sets from
the class, i.e., speaking loosely for the moment as to notation
(and the log case will not be defined or used again 
in this paper, but it is clear from 
context here what we mean by it; the exponential case 
will be defined only in Section~\ref{s:opt-body-version-TR}),
$
\indexedlwpp{1} \subseteq 
\indexedlwpp{2} \subseteq 
\indexedlwpp{3} \subseteq 
\dots \subseteq 
\indexedloglwpp  \subseteq 
\indexedpolylwpp  \subseteq 
\indexedexplwpp.$  As mentioned above, in this paper we will prove that 
the first five of these ``$\subseteq$''s are in fact all equalities.
We will also prove that 
the sixth ``$\subseteq$'' cannot be an equality unless 
NP is PP-low.

We now show that for every function $r$, $\indexedlwpp{r}$ is contained in the co-class
of the well-known counting class $\ceqp$.

\begin{definition}[\cite{sim:thesis:complexity,wag:j:succinct}]\label{d:ceqp-def}
 $\ceqp$ is the class of all sets $A$ such that there is a
  nondeterministic polynomial-time Turing machine $N$ and a function
  $f \in \fp$ such that for each $x\in\sigmastar$, 
  \[
     x \in A \Longleftrightarrow \accept_N(x) = f(x).
  \]
\end{definition} 

More convenient for us is the following characterization of $\ceqp$
using $\gapp$ functions.

\begin{theorem}[\cite{fen-for-kur:j:gap}]  \label{thm:ceqp-as-gap}
For each $A \subseteq \sigmastar$, $A \in \ceqp$ if and only if there exists a
function $g\in \gapp$ such that for all $x \in \sigmastar$,
\[
  x \in A \Longleftrightarrow g(x) = 0.
\]
\end{theorem}

We thus certainly have the following, which holds simply
by
taking the $\gapp$ function $g$ required in Theorem~\ref{thm:ceqp-as-gap} 
to be the same as the function $g$ in Definition~\ref{def:lwpp-indexed}.

\begin{theorem}
  For each function $r: \naturalnumbers \rightarrow \naturalnumbers$,
  $\indexedlwpp{r} \subseteq \coceqp$.
\end{theorem}

\section{Main Result:  LWPP Stays the Same If for Accepted Inputs We Allow
  Polynomially Many Gap Values Instead of One}\label{s:many}

We now state our main result: LWPP altered to 
allow even a polynomial number of 
target gap values is still LWPP (i.e., 
with just one target gap value).

\begin{theorem} \label{thm:lwpp-robust-poly}
$\indexedpolylwpp = \lwpp$.
\end{theorem}

For the proof, we need the following closure properties shown by
Fenner, Fortnow, and Kurtz~\cite{fen-for-kur:j:gap}.\footnote{%
We  mention in passing that, regarding 
Closure Property~\ref{closure4},
if the polynomial $q$ were allowed to have coefficients that 
are uncomputable---or that are extremely expensive to compute prefixes of 
the values of---real numbers, that claimed closure property 
might fail; we here are, as is typical in such settings, tacitly
assuming that the polynomials 
have rational 
coefficients.}
For function classes ${\mathcal F}_1$ and ${\mathcal F}_2$, 
${\mathcal F}_1 \circ {\mathcal F}_2 = 
\{ f_1 \circ f_2 ~|~ f_1\in {\mathcal F}_1 \land f_2 \in {\mathcal F}_2\}$,
where $\circ$ denotes composition.

\begin{closurep}[\cite{fen-for-kur:j:gap}]  \label{closure1}
   $\gapp \circ \fp = \gapp$ and $\fp \subseteq \gapp$.
\end{closurep}

\begin{closurep}[\cite{fen-for-kur:j:gap}]  \label{closure2}
   If $g \in \gapp$ then $-g \in \gapp$.
\end{closurep}

\begin{closurep}[\cite{fen-for-kur:j:gap}]   \label{closure4}
  If $g \in \gapp$ and $q$ is a polynomial, then the function
  \[
     h(x) = \prod_{0 \le i \le q(|x|)} g(\langle x, i \rangle)
   \]
   is in $\gapp$.
\end{closurep}

\begin{closurep}[\cite{fen-for-kur:j:gap}]   \label{coraddsubmult}
$\gapp$ is closed under addition, subtraction, and multiplication.
\end{closurep}

\begin{proof}[Proof of Theorem~\ref{thm:lwpp-robust-poly}]
As mentioned previously, it is easy to see that
$\lwpp \subseteq \indexedpolylwpp$.

To show $\indexedpolylwpp \subseteq \lwpp$, let $A$ be a set in
$\indexedpolylwpp$ defined by $g \in \gapp$, $f \in \fp$, 
and polynomial
$r(n)= n^c+c$ according to Definitions~\ref{def:lwpp-indexed} and~\ref{def:indexedpolylwpp}.

Let $h_1$ be a function such that for all $x\in \sigmastar$ and $i \in\positivenumbers$,
\[
  h_1(\langle x, i \rangle) =  f(\langle 0^{|x|}, i \rangle) -  g(x).
\]
We have $h_1 \in \gapp$ since $f \in\fp \subseteq \gapp$, $g \in\gapp$, and
$\gapp$ is closed under subtraction~\cite{fen-for-kur:j:gap}.
We define $h_2$ such that for all $x\in\sigmastar$,
\[
  h_2(x) = \prod_{1 \leq i \leq r(|x|)} h_1(\langle x, i \rangle).
\]
By Closure Property~\ref{closure4}, $h_2\in \gapp$.
Note that for all $x \in \sigmastar$,
\[
  h_2(x) = \prod_{1 \leq i \leq r(|x|)} \left( f(\langle 0^{|x|}, i \rangle) - g(x) \right).
\]
It follows that for every $x\in \sigmastar$,
\begin{equation}  \label{eq:h2}
  h_2(x) = \begin{cases}
                 0 & \text{if there exists } i \in \{ 1, 2,
  \ldots , r(|x|)\} \text{ such that } g(x) = f(\langle 0^{|x|}, i \rangle)\\
                 \prod_{1 \leq i \leq r(|x|)}  f(\langle 0^{|x|}, i \rangle)  & \text{if } g(x) = 0. 
               \end{cases}
\end{equation}
Now we define the function $\widehat{g}$ such that for all $x \in\sigmastar$,
\[
  \widehat{g}(x) = h_2(x) - \prod_{1 \leq i \leq r(|x|)} f(\langle 0^{|x|}, i \rangle)  .
\]
Using the closure properties, it is easy to see that $\widehat{g} \in \gapp$.

Note that by Eqn.~(\ref{eq:h2}), we have that for all $x \in\sigmastar$,
\begin{equation}   \label{eq:widehatg}
\widehat{g}(x) = \begin{cases}  
               - \prod_{1 \leq i \leq r(|x|)} f(\langle 0^{|x|}, i \rangle)   &  \text{if there exists } i \in \{ 1, 2,
  \ldots , r(|x|)\} \text{ such that } g(x) = f(\langle 0^{|x|}, i \rangle)\\
                0  & \text{if } g(x) = 0. 
               \end{cases}
\end{equation}
Let $\widehat{f}$ be a function such that for all $\ell\in\naturalnumbers$,
\[
  \widehat{f}(0^\ell) =  - \prod_{1 \leq i \leq r(\ell)} f(\langle 0^{\ell}, i \rangle) .
\]
It is easy to see that $\widehat{f} \in \fp$.
Keeping in mind Eqn.~(\ref{eq:widehatg}), it follows from the above
that for every $x\in\sigmastar$,
\[
  \widehat{g}(x) = \begin{cases}
               \widehat{f}(0^{|x|})   & \text{if } x \in A \\
                0  & \text{otherwise}. 
               \end{cases}
\] 
Since $\widehat{g} \in \gapp$ and $\widehat{f} \in \fp$, this implies that $A \in \lwpp$.
\end{proof}

A theorem analogous to Theorem~\ref{thm:lwpp-robust-poly} also holds
for the corresponding class that allows the target values to
depend on the actual input instead on only the length of the input.

\begin{definition}
\[
   \indexedpolywpp = \bigcup_{c \in
     \positivenumbers} \indexedwpp{(n^c+c)}.
\]
\end{definition}

\begin{theorem} \label{thm:wpp-robust-poly}
$\indexedpolywpp = \wpp$.
\end{theorem}
The proof is 
almost exactly the same 
as the proof of
Theorem~\ref{thm:lwpp-robust-poly}, simply taking into 
account the fact that for $\wpp$ the ``gap'' function 
can vary even among inputs of the same length.

We remark that the robustness results stated in Theorems~\ref{thm:lwpp-robust-poly}
and~\ref{thm:wpp-robust-poly} do not seem to in any obvious way follow
as corollaries to the closure of LWPP under polynomial-time Turing 
reductions~\cite{fen-for-kur:j:gap} or of WPP under polynomial-time truth-table reductions~\cite{spa-tha-rah:j:quantum}.

Let us now see whether we can extend Theorems~\ref{thm:lwpp-robust-poly}
and~\ref{thm:wpp-robust-poly}.
Suppose we not only allow the target-gap set for acceptance to 
have polynomially many values, but in addition allow 
the target-gap set for rejection to 
have polynomially many values.  
(In
contrast, both $\lwpp$ and $\indexedlwpp{r}$ allow rejection only when
the gap's value is 0).  
Is $\lwpp$ so 
robust that even \emph{that} class is no larger than $\lwpp$?  
We will now build on our main result to give the answer ``yes'' to 
that question, thus extending our main result to this more symmetric
case for $\lwpp$ and for $\wpp$.

To this end, let us define
$\indexedlwpp{(r_A,r_R)}$ as follows.
(One of course can in the clear, analogous way define a 
similarly loosened
version of $\wpp$, 
$\indexedwpp{(r_A,r_R)}$.)

\begin{definition}  \label{def:rArR-LWPP}
  Let $r_A$ and $r_R$ be any functions mapping from $\naturalnumbers$
  to $\naturalnumbers$.  Then $\indexedlwpp{(r_A,r_R)}$ is
  the class of all sets $B$ such that there exists a $\gapp$ function
  $g$ and functions---each mapping to $\integers$---$f_A \in \fp$ and
  $f_R \in \fp$ such that both of the following hold:
\begin{enumerate}
\item For each $j \in \naturalnumbers$, $A_j \cap R_j = \emptyset$,
  where
  $$A_j = \{ n ~ | ~ (\exists i \in \{1,2,\dots,r_A(j)\})                       
  [f_A(\langle 0^j,i \rangle) = n] \}$$ and
  $$R_j = \{ n ~ | ~ (\exists i \in \{1,2,\dots,r_R(j)\})                       
  [f_R(\langle 0^j,i \rangle) = n] \}.$$
\item
For each $x\in\sigmastar$,
\begin{align*}                                                                  
  x \in B & \Longrightarrow g(x) \in A_{|x|}                                    
\\                                                                              
  x \not\in B & \Longrightarrow g(x) \in R_{|x|}.                               
\end{align*}
\end{enumerate}
\end{definition}
It is not hard to see, via shifting gaps with dummy paths, that the
class
$\indexedlwpp{(r,1)}$ equals $\indexedlwpp{r}$.  However, what can be
shown more generally about the classes $\indexedlwpp{(r_A,r_R)}$?  For
example, for which functions $r_{1,A}$, $r_{1,R}$, $r_{2,A}$, and
$r_{2,A}$ does it hold that
$(r_{1,A},r_{1,R})\mhyphen\lwpp \subseteq
(r_{2,A},r_{2,R})\mhyphen\lwpp$?  (There are some literature notions
that, in different ways, have at least a somewhat similar flavor to
our notion, namely, defining classes by being more flexible regarding
acceptance types.  In particular, the counting classes ${\rm CP}_S$ of
Cai et al.~\cite{cai-gun-har-hem-sew-wag-wec:j:bh2} and the
``C-class'' framework of Bovet, Crescenzi, and
Silvestri~\cite{bov-cre-sil:j:uniform} have such a flavor. But in
contrast with those, our classes here are ones whose definitions are
centered on the notion of gaps.) We do not here undertake 
that general study, but instead resolve what seems the 
most compelling question, namely, we prove that 
$\indexedlwpp{(\mathit{Poly},\mathit{Poly})}  = \lwpp$.

\begin{definition} \label{def:polypolylwpp}
$   \indexedlwpp{(\mathit{Poly}, \mathit{Poly})} = \bigcup_{c \in
     \positivenumbers} \indexedlwpp{(n^c+c, n^c+c)}.$
\end{definition}
\begin{theorem} \label{thm:polypolylwpp-in-polylwpp}
$\indexedlwpp{(\mathit{Poly},\mathit{Poly})} = \indexedpolylwpp$.
\end{theorem}
The proof of Theorem~\ref{thm:polypolylwpp-in-polylwpp} can be
found in the appendix.
Together with Theorem~\ref{thm:lwpp-robust-poly},
we get the following corollary.
\begin{corollary}
$\indexedlwpp{(\mathit{Poly},\mathit{Poly})}  = \lwpp$.
\end{corollary}

An analogous statement can also be shown for the analogously defined class 
$\indexedwpp{(\mathit{Poly},\mathit{Poly})}$:
\begin{theorem}
$\indexedwpp{(\mathit{Poly},\mathit{Poly})}  = \wpp$.
\end{theorem}

\section{Applying the Main Result to Graph Reconstruction}\label{s:app}

\begin{definition}
Let $\langle G_1, G_2, \ldots , G_n\rangle$ be a sequence of graphs and $G = (V,
E)$ a graph with $V = \{ 1, 2, \ldots , n\}$.  Suppose that there is a
permutation $\pi \in S_n$ such that for each $k \in \{1, 2, \ldots ,
n\}$, the graph $G_{\pi(k)}$ is isomorphic to the graph
$(V-\{k\} , E - \{ \{ k, \ell\} : \ell \in V\})$ obtained by deleting
vertex $k$ from $G$.
Then $\langle G_1, G_2, \ldots , G_n\rangle$ is called a {\em deck} of $G$ and $G$
is called a {\em preimage} of the sequence $\langle G_1, G_2, \ldots , G_n\rangle$.

A sequence of graphs $\langle G_1, G_2, \ldots , G_n\rangle$ is called a {\em legitimate
  deck} if there exists a graph $G$ that is a preimage of 
$\langle G_1, G_2, \ldots , G_n \rangle$.
\end{definition}

The {\em Reconstruction Conjecture} 
(see,
e.g., the surveys~\cite{bon-hem:j:reconstruction,man:j:reconstruction-progress-survey,bon:b:graph-recon} and 
the book~\cite{lau-sca:b:graph-reconstruction})
says that each legitimate deck
consisting of graphs with at least two vertices has
exactly one preimage up to isomorphism. This conjecture is a very
prominent conjecture in graph theory---as mentioned in
Section~\ref{s:intro} it is perhaps the most important conjecture in 
that area---and has been
studied for many decades.
 
Nash-Williams~\cite{nas:b:sel-top-graph}, 
Mansfield~\cite{mansfield:j:led}, Kratsch and 
Hemachandra~\cite{hem-kra:j:legitimate-deck},
and Hemaspaandra et al.~\cite{hem-hem-rad-tri:j:reconstruction-2}
 introduced
various decision problems related to the Reconstruction Conjecture,
as part of a stream of work studying the algorithmic and complexity 
issues of reconstruction.
We here are interested mainly in the Legitimate Deck Problem, which is
defined as the following decision problem.

\medskip

\noindent
{\bf Legitimate Deck (a.k.a.~the Legitimate Deck Problem)~\cite{mansfield:j:led}:}  Given a sequence of
graphs $\langle G_1, G_2, \ldots , G_n\rangle$,
is $\langle G_1, G_2, \ldots , G_n\rangle$ a legitimate deck?

\medskip

Mansfield~\cite{mansfield:j:led}
showed that the Graph
Isomorphism Problem, GI (given two graphs $G_1$ and $G_2$, are they
isomorphic?),~is polynomial-time many-one reducible to the Legitimate Deck Problem.
However, to this day it 
remains open whether there is a 
polynomial-time many-one 
reduction from the Legitimate
Deck Problem to GI\@.  

So how hard is the Legitimate Deck Problem?  It is easy to see
that the Legitimate Deck Problem is in $\np$.  In the following, we will see that
there is some evidence that the Legitimate Deck Problem 
is not NP-hard, and we will
improve that evidence.

Let us define the following function problem.

\medskip

\noindent
{\bf Preimage
  Counting~\cite{hem-kra:cOUTDATED-HAS-APPEARED:legitimate-deck,hem-kra:j:legitimate-deck}:}
Given a sequence of graphs $\langle G_1, G_2, \ldots , G_n\rangle$,
 compute the number $\pcount(\langle G_1, G_2, \ldots , G_n \rangle )$
of all nonisomorphic preimages for the sequence $\langle G_1, G_2,
\ldots , G_n \rangle$.

\medskip

K\"{o}bler, Sch\"{o}ning, and Tor\'{a}n~\cite{koe-sch-tor:j:gi-pplow} showed the following theorem.

\begin{theorem}[\cite{koe-sch-tor:j:gi-pplow}] \label{thm:kst}
  There is a function $h: 0^* \rightarrow \naturalnumbers - \{ 0 \}$ in $\fp$
 such that the function that maps
  every sequence $\langle G_1, G_2, \ldots , G_n \rangle$ to $h(0^n)$ times
  the number of nonisomorphic preimages, i.e., 
\[
  \langle G_1, G_2, \ldots , G_n \rangle \mapsto h(0^n) \cdot 
\pcount( \langle G_1, G_2, \ldots , G_n \rangle ) 
\]
is in $\gapp$.
\end{theorem}

Theorem~\ref{thm:kst} has the following corollary.

\begin{corollary}[\cite{koe-sch-tor:j:gi-pplow}]  \label{thm:reco-implied-ld-in-lwpp}
 If the Reconstruction Conjecture holds, then the Legitimate Deck Problem is in
 $\lwpp$.
\end{corollary}

Since all LWPP sets are PP-low, that immediately gives some
evidence that the Legitimate Deck Problem is not NP-hard.

\begin{corollary}\label{c:low-1}
 If the Reconstruction Conjecture holds, then the Legitimate Deck Problem is 
not $\np$-hard (or even $\np$-Turing-hard) unless $\np$ is $\pp$-low.
\end{corollary}

Unfortunately, we do not know 
whether the Reconstruction Conjecture holds.  However,
perhaps we can prove the membership of the Legitimate Deck Problem in LWPP under a
\emph{weaker assumption} than the Reconstruction Conjecture, for instance,
what if
the number of nonisomorphic preimages is not always 0 or 1 (as holds
under the Reconstruction Conjecture), but rather is merely relatively small,
e.g., some constant or some polynomial in the number of vertices (as
must hold for each graph class 
having bounded minimum degree)?  We will now use the 
results of the previous section to prove that this indeed is the case.

\begin{conjecture}[$q$-Reconstruction Conjecture]
For each legitimate deck there exist at most $q(n)$ nonisomorphic
preimages,
where $n$ is the number of graphs in the deck.
\end{conjecture}
For any function $r$, we now define a complexity class $\multipliedlwpp{r}$.
This class may not seem very natural, but we will see that it
is very-well suited to helping us classify the problem $\ld$.  In 
some sense, it is a tool that we will use in our proof, and then 
will discard by noting that it in fact turns out to be a 
disguised version of
$\indexedlwpp{r}$.

\begin{definition} \label{def:rlwpp-with-multiplied} %
Let $r$ be any function mapping from $\naturalnumbers$ to
$\naturalnumbers$.
Then the class
$\multipliedlwpp{r}$ is
the class of all sets $A$ such that there exists a $\gapp$
  function $g$, and a function $f \in \fp$ that maps from $0^*$ to
  $\integers - \{ 0\}$,  
such that for each $x\in\sigmastar$,
\begin{align*}
  x \in A & \Longrightarrow \text{there exists } i \in \{ 1, 2,
  \ldots , r(|x|)\} \text{ such that } g(x) = i \cdot f(0^{|x|})\\
  x \notin A & \Longrightarrow g(x) = 0.
\end{align*}
\end{definition}

\begin{theorem}  \label{thm:qreco-implies-qlwpp}
Let $q$ be any nondecreasing function from $\naturalnumbers$ to
$\naturalnumbers$. Then the following holds.
  
 If the $q$-Reconstruction Conjecture holds, then the Legitimate Deck Problem is in
 $\multipliedlwpp{q}$.
\end{theorem}

\begin{proof} %
Suppose that the  $q$-Reconstruction Conjecture holds.
Let $\langle G_1, G_2, \ldots , G_n \rangle$ be an input to the Legitimate Deck
Problem. By our assumption, we have that
\begin{align*}
  \langle G_1, G_2, \ldots , G_n \rangle \in \ld & \Longrightarrow
       \pcount(\langle G_1, G_2, \ldots , G_n\rangle ) \in \{ 1, 2, \ldots , q(n)\}, \text{and} \\
  \langle G_1, G_2, \ldots , G_n \rangle \notin \ld & \Longrightarrow
       \pcount(\langle G_1, G_2, \ldots , G_n\rangle ) = 0. 
\end{align*}
Let $h \in\fp$ be the function discussed in Theorem~\ref{thm:kst}.
Then the function $g$ defined such that for every sequence
$\langle G_1, G_2, \ldots , G_n \rangle$,
\[
  g(\langle G_1, G_2, \ldots , G_n \rangle) =  h(0^n) \cdot \pcount(
  \langle G_1, G_2, \ldots , G_n\rangle )  
\]
is in $\gapp$.
It follows that
\begin{align*}
  \langle G_1, G_2, \ldots , G_n \rangle \in \ld & \Longrightarrow
       g(\langle G_1, G_2, \ldots , G_n\rangle) = h(0^n) \cdot i \text{ for some }\\
           & \quad \quad\quad i\in \{ 1, 2, \ldots , q(n)\}, \text{ and} \\
  \langle G_1, G_2, \ldots , G_n \rangle  \notin \ld & \Longrightarrow
       g(\langle G_1, G_2, \ldots , G_n\rangle ) = 0. 
\end{align*}
This almost directly implies that $\ld \in \multipliedlwpp{q}$.
However, note that to satisfy Definition~\ref{def:rlwpp-with-multiplied},
function $h$ must be a function 
\emph{that depends only on the length of the input
$\langle G_1, G_2, \ldots , G_n\rangle$}.  The problem here is that
(depending on how exactly we decide to encode the graphs $G_1, G_2,
\ldots , G_n$ in the input $\langle G_1, G_2, \ldots , G_n\rangle$)
even the value $n$ (the number of graphs in the deck)
may depend on the actual input $\langle G_1, G_2, \ldots , G_n\rangle$ and not only
on the length of the input $\langle G_1, G_2, \ldots , G_n\rangle$.  

Fortunately, there is a way to get around this problem. 
From the length of the input, we get at least an upper bound for $n$
since we can certainly assume that
\[
  n \le |\langle G_1, G_2, \ldots , G_n\rangle |.
\]
Define a function $\widehat{h}$ such that for each $m \in\naturalnumbers$,
$\widehat{h}(0^m)$ is the product of all ``h-values'' \emph{up to length
  m}.
That is, define function $\widehat{h}$ such that for all $m \in
\naturalnumbers$,
\[
  \widehat{h}(0^m) = \prod_{0 \leq i \leq m} h(0^i).
\]
Define function $h'$ such that for all inputs $\langle G_1, G_2, \ldots , G_n\rangle$,
\[
  h'(\langle G_1, G_2, \ldots , G_n\rangle) = \prod_{0 \leq i \leq
          |\langle G_1, G_2, \ldots , G_n\rangle|  \,\land\, i \neq n}  h(0^i).
\]
Now we can see that for all inputs $\langle G_1, G_2, \ldots , G_n\rangle$,
\begin{multline} \label{eq:big-equation}
  g(\langle G_1, G_2, \ldots , G_n\rangle) \cdot h'(\langle G_1, G_2,
  \ldots , G_n\rangle) \\
   = 
  \begin{cases}
    0 & \parbox{1.5in}{$\text{if } \langle G_1, G_2, \ldots , G_n\rangle  \notin$\\$\ld$}\\[10pt]
    i \cdot \widehat{h}(0^{| \langle G_1, G_2, \ldots , G_n\rangle |}) \text{ for some } i \in \{ 1, 2, \ldots
    , q(n)\} & \parbox{1.5in}{$\text{if } \langle G_1, G_2, \ldots , G_n\rangle  \in$\\$\ld.$}
  \end{cases}
\end{multline}
Note that $\widehat{h}, h' \in \fp$ and $g \in \gapp$.
By Closure Properties~\ref{closure1} and~\ref{coraddsubmult}, the function
\[
  x \mapsto g(x) \cdot h'(x)
\]
is in $\gapp$.
Finally, note that since $q$ is nondecreasing  and
$n \le |\langle G_1, G_2, \ldots , G_n\rangle|$, we have that
$q(n) \le q(|\langle G_1, G_2, \ldots , G_n\rangle|)$ in
Eqn.~(\ref{eq:big-equation}).  Thus by 
Definition~\ref{def:rlwpp-with-multiplied}---with 
the $\gapp$ function $g$ there being our $g(x) \cdot h'(x)$
and the function $f$ there being our $\widehat{h}$---we have that 
$\ld \in \multipliedlwpp{q}$.
\end{proof}

We want to show that if there exists a polynomial $q$ such that the $q$-Reconstruction Conjecture holds,
then $\ld$ is in $\lwpp$.  To this end, we need 
the following 
inclusion.

\begin{theorem}  \label{thm:multiplied-in-indexed}
  For each function $r: \naturalnumbers \rightarrow \naturalnumbers$, $\multipliedlwpp{r} \subseteq \indexedlwpp{r}$.
\end{theorem}
\begin{proof}
Let $A$ be a set in $\multipliedlwpp{r}$ via $f_1\in \fp$ and $g\in\gapp$.
Let $f_2 \in \fp$ be the function defined such that for every $n, i \in
\positivenumbers$, $f_2(\langle 0^n, i \rangle) = i\cdot f_1(0^n)$.  Then $A$ is in
$\indexedlwpp{r}$ via $f_2\in \fp$ and $g\in \gapp$.
\end{proof}   

\begin{corollary}
  If the $q$-Reconstruction Conjecture holds for some polynomial $q$, then the Legitimate Deck Problem is in
 $\lwpp$.
\end{corollary}
\begin{proof}
Suppose the  $q$-Reconstruction Conjecture holds for nondecreasing polynomial $q$.   
(If $q$ is not nondecreasing but the 
$q$-Reconstruction Conjecture holds, then obviously we can replace 
$q$ with a nondecreasing polynomial $q'$ that on each input $n$
is greater than or equal to $q(n)$ and the 
$q'$-Reconstruction Conjecture will hold. So we may w.l.o.g.~take 
it that $q$ is nondecreasing.)
By Theorems~\ref{thm:qreco-implies-qlwpp} and~\ref{thm:multiplied-in-indexed}, 
$\ld \in \multipliedlwpp{q} \subseteq \indexedlwpp{q}$ and hence $\ld \in \indexedpolylwpp$.
With Theorem~\ref{thm:lwpp-robust-poly}, it follows that $\ld\in\lwpp$.
\end{proof}

This gives us our new, more flexible---though still
conditional---evidence that the Legitimate Deck Problem 
is not NP-hard.

\begin{corollary}\label{c:low-q}
 If the $q$-Reconstruction Conjecture holds for some polynomial $q$, 
then the Legitimate Deck Problem is 
not $\np$-hard (or even $\np$-Turing-hard) unless $\np$ is $\pp$-low.
\end{corollary}

\begin{definition}[\cite{hem-kra:j:legitimate-deck}]
Let $\mathcal{H}$ be any class of graphs.
Then the {\em Legitimate Deck Problem restricted to $\mathcal{H}$}
  consists of all sequences of graphs $\langle G_1, G_2, \ldots , G_n
  \rangle$ such that $\langle G_1, G_2, \ldots , G_n \rangle$
 is a legitimate deck and for each $i\in \{ 1, 2, \ldots , n\}$, $G_i$ is 
in $\mathcal{H}$.
\end{definition}
Note that the above definition is so flexible that it allows even the
case where the preimage(s) may not be in $\mathcal{H}$.

\begin{theorem} \label{thm:restricted-graph-class}
Let $\mathcal{H}$ be any P-recognizable class of graphs such that
decks
consisting only
of graphs in $\mathcal{H}$ have a number of nonisomorphic preimages that is bounded polynomially
in the number of graphs in the deck.  Then the Legitimate Deck Problem restricted to 
$\mathcal{H}$ is in $\lwpp$.
\end{theorem}
\begin{proof}
Let $\mathcal{H}$ be any P-recognizable class of graphs
and $q$ a nondecreasing polynomial
 such that
decks consisting only
of graphs in $\mathcal{H}$ have a number of nonisomorphic preimages
that is bounded by $q(n)$, where $n$ is the number of graphs in the deck.

First, we show that the  Legitimate Deck Problem restricted to $\mathcal{H}$
 is in $\multipliedlwpp{q}$.  
Let $\langle G_1, G_2, \ldots , G_n
  \rangle$ be an input to the  Legitimate Deck Problem.
Check if for every $i \in  \{ 1, 2, \ldots , n\}$, $G_i$ is 
in $\mathcal{H}$.  If this is not the case then reject in the sense of
$\multipliedlwpp{q}$, i.e., produce a gap of zero.
Otherwise, the deck $\langle G_1, G_2, \ldots , G_n
  \rangle$ has at most $q(n)$ preimages, i.e., $\pcount(\langle G_1, 
G_2, \ldots , G_n \rangle \le q(n)$.  Proceed as in the proof 
of Theorem~\ref{thm:qreco-implies-qlwpp}.

Since by Theorems~\ref{thm:multiplied-in-indexed}
and~\ref{thm:lwpp-robust-poly}, 
$\multipliedlwpp{q} \subseteq \indexedlwpp{q} \subseteq
\lwpp$, it follows that  the  Legitimate Deck
Problem restricted to $\mathcal{H}$ is in $\lwpp$.
\end{proof}

As usual, for each graph $G$, $\delta(G)$ denotes the degree of a
minimum-degree vertex of $G$.

\begin{theorem}  \label{thm:graphs-with-bounded-mindegree}
For each $k \in \positivenumbers$, let 
\[
  \mathcal{H}_k = \{ G \, | \, G\text{ is a graph such that }\delta(G)  \le k \}.
\]
Then the Legitimate Deck Problem restricted to $\mathcal{H}_k$ is in $\lwpp$.
\end{theorem}

\begin{proof}
In the proof of their Theorem 6.1, Kratsch and Hemaspaandra~\cite{hem-kra:j:legitimate-deck}
showed that for each class of graphs with bounded minimum degree, the
number of nonisomorphic preimages is polynomially bounded.
Now the theorem follows from Theorem~\ref{thm:restricted-graph-class}.
\end{proof}
It is interesting to note that for each of the $\mathcal{H}_k$
classes $\gi_{\mathcal{H}_k} \equiv^p_m \gi$
trivially holds
(for example, via adding to each of the two graphs being 
tested for isomorphism an isolated node),
notwithstanding 
the fact that intersection
with $\mathcal{H}_k$
pulls 
the Legitimate Deck Problem's complexity into $\lwpp$.

In two of this section's
corollaries we used the fact that all LWPP sets are PP-low.
We mention that since all LWPP set are 
also $\ceqp$-low~\cite{koe-sch-tor:j:gi-pplow,fen-for-kur:j:gap}, 
the altered versions of Corollaries~\ref{c:low-1} 
and~\ref{c:low-q} in which 
the conclusion is changed from 
``unless $\np$ is $\pp$-low''
to
``unless $\np$ is $\ceqp$-low''
both hold, respectively due to 
K\"{o}bler, Sch\"{o}ning, and Tor\'{a}n~\cite{koe-sch-tor:j:gi-pplow}
and the present paper.

\section{Optimality of the Main Result}\label{s:opt-body-version-TR}
%
%
%
%
%
%
%
%
%
%
%
%
%
%
%
%
%
%
%
%
%
%
%
%
%
%
%
%
%
%
%
%
%
%
%
%
%
%
%
%
%
%
%
%
%
%
%
%
%
%
%
%
%
%
%
%

It is easy to see that the proof of Theorem~\ref{thm:lwpp-robust-poly}
breaks down if 
$\indexedpolylwpp$ is replaced by the analogous class
where the size of the set of allowed gap values can be larger than polynomial in
the input length.  
This does not necessarily imply that the
corresponding theorem does not hold.  However, in this section, we 
establish that
relativizable proof techniques are not sufficient to improve 
Theorem~\ref{thm:lwpp-robust-poly} from $\indexedpolylwpp$ to 
$\indexedlwpp{r}$ for 
any function $r$ that is not polynomially bounded.
That is, we show that our main result is optimal with respect to 
what can be proven by relativizable proof techniques.

But first, let us briefly 
consider the ``more extreme'' case of the class $\indexedexplwpp$,
where the number of allowed gap values can be an exponential function.
\begin{definition}
\[
  \indexedexplwpp = \bigcup_{c\in \positivenumbers} \indexedlwpp{2^{n^c+c}} .
\] 
\end{definition}
We show that the whole class $\np$ is contained in in $\indexedexplwpp$.  
This gives strong
evidence that $\indexedexplwpp \not\subseteq \lwpp$ because $\lwpp$ is known
to be low for $\pp$~\cite{fen-for-kur:j:gap}, but
$\np$ is widely believed not to be low for $\pp$.
\begin{theorem}  \label{thm:coceqp-contained}
  $\coceqp \subseteq \indexedexplwpp$ (and hence---since 
 $\coceqp \supseteq \indexedexplwpp$ is immediate from the 
definitions---$\coceqp = \indexedexplwpp$).
\end{theorem}
Since $\np \subseteq \coceqp$ and all $\lwpp$ sets are PP-low, we 
obtain the following corollaries.
\begin{corollary}
   $\np \subseteq \indexedexplwpp$.
\end{corollary}

\begin{corollary}
   If $\indexedexplwpp = \lwpp$ then $\np \subseteq \lwpp$ and $\rm PP^{NP} = PP$ (and, indeed, $\coceqp \subseteq \lwpp$ and $\rm PP^{\coceqp} = PP$).
\end{corollary}

All of the above, except the PP-lowness claims, analogously 
hold for $\indexedexpwpp$/$\wpp$, since $r$-LWPP is a subset of $r$-$\wpp$.  
So for example we have the following.

\begin{corollary}
   If $\indexedexpwpp = \wpp$ then $\np \subseteq \wpp$
 (and, indeed, $\coceqp \subseteq \wpp$).
\end{corollary}

\begin{proof}[Proof of Theorem~\ref{thm:coceqp-contained}]
Let $A \in \coceqp$. 
By Theorem~\ref{thm:ceqp-as-gap}, there exists a function $g \in \gapp$ such that for every $x\in
\sigmastar$, 
\begin{align*}
   x & \in A \Longrightarrow g(x) \not= 0, \text{ and}\\
   x & \notin A \Longrightarrow g(x) = 0.
\end{align*}
Let $p$ be a polynomial such that for each $x\in\sigmastar$, $-2^{p(|x|)}\le g(x) \le 2^{p(|x|)}$.
Let $f\in \fp$ be the function such that, for every $n
\in \positivenumbers$ and every $i \in \positivenumbers$,
\[
  f(\langle 0^n, i \rangle) = \begin{cases}
           i/2 & \text{if } i \text{ is even}\\
          -(i+1)/2 & \text{otherwise.}
              \end{cases}
\]
We can now see that according to Definition~\ref{def:lwpp-indexed}, $A \in 
\indexedlwpp{2^{p(n)+1}}$ and hence $A \in \indexedexplwpp$.
\end{proof}

%

\begin{theorem} \label{thm:lower-bound}
  Let $r$ be any function from $\naturalnumbers$ to $\naturalnumbers$
such that for every $c \in \naturalnumbers$,
  $r \notin \bigoh(n^c)$.  Then there exists an oracle
$\mathcal{O}$
such that $\indexedlwpp{r}^\mathcal{O} \not\subseteq \lwpp^\mathcal{O}$.
\end{theorem}

To prove Theorem~\ref{thm:lower-bound}, we will encode nondeterministic
oracle Turing machines by low-degree multivariate polynomials.  This
technique is apparently folklore and has been 
used, for example, by
de~Graaf and Valiant~\cite{gra-val:t:eqp} to construct a relativized world 
where the quantum complexity class $\eqp$ 
is not contained in the modularity-based complexity class $\MOD_{p^k}\pe$.
The general 
technique 
of replacing oracle machines by simpler combinatorial
objects such as circuits, decision trees, or polynomials and then using
properties of such combinatorial objects to show the existence of a
desired oracle 
dates to the seminal
work of Furst, Saxe, and 
Sipser~\cite{fur-sax-sip:j:parity}, who made the connection between circuit
lower bounds and the relativization of the polynomial hierarchy---a 
connection that has led to the resolution of many previously long-open 
relativized questions, such as the achievement of an oracle 
making the polynomial hierarchy infinite~\cite{yao:c:separating,has:thesis:small-depth} and of oracles making the polynomial hierarchy 
extend exactly $k$ levels~\cite{ko:j:exact}.

As indicated above, 
in Definition~\ref{def:poly-encoding} and Proposition~\ref{prop:poly-encoding} below, we
will 
encode nondeterministic
oracle Turing machines by low-degree multivariate polynomials,
in order to 
obtain polynomials that compute the gap of polynomial-time
nondeterministic oracle machines.

\begin{definition}[Folklore; see also~\cite{spa-tha-rah:j:quantum}]  
\label{def:poly-encoding}
Let $N^{(\cdot)}(x)$ be a nondeterministic polynomial-time oracle Turing machine
with running time $t(.)$  and input $x\in\sigmastar$.
Let $x_1, x_2, \ldots, x_{m}$ be an enumeration of all strings in
$\sigmastar$ up to length $t(|x|)$.

A {\em polynomial encoding of
  $N^{(\cdot)}(x)$\/} is a multilinear polynomial 
$p \in \integers[y_1, y_2, \ldots, y_{m}]$ defined as follows:
Call a computation path $\rho$ {\em valid\/} if $\rho$
is a computation path of $N^{D}(x)$ for some oracle $D\subseteq\sigmastar$.
Let $x_{i_1}, x_{i_2}, \ldots, x_{i_{\ell}}$ be the distinct queries
along a valid computation path $\rho$. Create a monomial $\mono(\rho)$ that is the product
of terms
$z_{i_k}$, $k \in \{ 1, 2, \ldots,  \ell\}$, where
$z_{i_k} = y_{i_k}$ if $x_{i_k}$ is answered ``yes'' and $z_{i_k} = (1-y_{i_k})$ if $x_{i_k}$ is answered ``no'' along $\rho$. Define
\[
p(y_1, y_2, \ldots, y_m) = \sum_{\rho: \rho \textnormal{ is valid}}\sign(\rho) \cdot \mono(\rho).
\]
Here, $\sign(\rho) = 1$ if $\rho$ is an accepting path and
$\sign(\rho) = -1$ if $\rho$ is a rejecting path.
\end{definition}

The next proposition states that the multilinear polynomial $p$ has low total
degree, and contains all
the necessary information about $N^{(\cdot)}(x)$ to yield the value $\gap_{N^{{B}}}(x)$ for every
oracle $B\subseteq \sigmastar$.

\begin{proposition}[Folklore; see also~\cite{spa-tha-rah:j:quantum}] \label{prop:poly-encoding}
The polynomial $p(y_1, y_2, \ldots, y_m)$ 
defined in Definition~\ref{def:poly-encoding}
has the following properties:
\begin{enumerate}
\item $\degree(p) \leq t(|x|)$, and
\item for each $B \subseteq \sigmastar$, $p(\chi_{B}(x_1), \chi_{B}(x_2), \ldots, \chi_{B}(x_{m}))$ $=$
$\gap_{N^B}(x)$.
\end{enumerate}
\end{proposition}

\begin{lemma}[\cite{spa-tha-rah:j:quantum}]
\label{lemma:prime-divisor}
Let $N, p \in \naturalnumbers$ be such that $p$ is a prime and $p \leq N/2$.
Let $s \in \integers[y_1, y_2, \ldots, y_N]$ be a multilinear polynomial with
total degree $\degree(s) < p$. If
for some $val \in \integers$, it holds that
\begin{enumerate}
\item $s(0, 0, \ldots, 0) = 0$, and 
\item $s(y_1, y_2, \ldots, y_{N}) = val$, for every $y_1, y_2, \ldots, y_N \in \{0,1\}$ with $p = \sum_{1 \leq i \leq N} y_i$, 
\end{enumerate}
then $p \mid val$.
\end{lemma}

Lemma~\ref{lemma:prime-number-theorem} states a variant of the prime number
theorem. We will need it in our oracle construction to get a lower
bound for  the number
of primes in a given set of natural numbers (see set $S$ defined below).

\begin{lemma}[\cite{ros-sch:j:prime}]
\label{lemma:prime-number-theorem}
For every $n \geq 17$, the number of primes less than or equal to $n$, $\pi(n)$, satisfies
\[
   \pi(n) > n/\ln n.
\]
\end{lemma}

\begin{proof}[Proof of Theorem~\ref{thm:lower-bound}]
First, we need a test language.

For every set $B \subseteq \sigmastar$, define $L_B$ as
\[
  L_B = \{ 0^n \, | \, \|B^{=n}\| > 0 \},
\]
where $B^{=n}$ denotes $B \cap \Sigma^n$.
If $B$ satisfies the condition 
that for every length $n$ it holds that $\|B^{=n}\| \le r(n)$,
then $L_B \in \indexedlwpp{r}$.  
(To see this, note that the function $d(0^n) = 
\|B^{=n}\|$ is a $\sharpp^B$ function, where $\sharpp$ is 
Valiant's~\cite{val:j:permanent}
class of functions that count the number of \emph{accepting} 
paths of nondeterministic polynomial-time Turing machines.
But 
$\sharpp \subseteq \gapp$~\cite{fen-for-kur:j:gap},
and that fact itself relativizes, so $d(0^n) \in \gapp^B$.
Since our test language should reject when $d(0^n) = 0$
and should accept when $1 \leq d(0^n) \leq r(n)$, and 
by our ``if $B$ satisfies'' those are the only possibilities, we
have $L_B \in \indexedlwpp{r}^B$.)

We will construct an oracle $B$ such that for each $n$, $\|B^{=n}\| \le r(n)$
and $L_B \notin \lwpp^B$.  
Let $( N_j, M_j, p_j)_{j \ge 1}$ be an enumeration of all triples such that
$N_j$ is a nondeterministic polynomial-time oracle Turing machine, 
$M_j$ is a deterministic polynomial-time oracle Turing machine
computing a function, and $p_j$
is a monotonically increasing polynomial such that the running time of both $N_j$ and $M_j$ is
bounded by $p_j$ regardless of the oracle.  We construct the oracle $B$ in
stages.  In stage $j$, we decide the membership in $B$ of strings
of length $n_j$ and extend the initial segment $B_{j-1}$ of $B$ 
to $B_j$.  Initially, we set $B_0 = \emptyset$.

\smallskip

\noindent
{\boldmath\bf Stage $j$, where $j \ge 1$:} 
Let $n_j$ be large enough that:
(a)~$n_j > p_{j-1}(n_{j-1})$ (to ensure that the previous stages 
are not affected),  
(b)~$r(n_j) \geq p_j(n_j)^4$,
(c)~$(2^{n_j} - p_j(n_j))/2 \geq p_j(n_j)^4$, and 
(d)~$p_j(n_j)^3 - p_j(n_j) \geq p_j(n_j)^2$.  
Such an $n_j$ exists because $p_j$ is a monotonically increasing
polynomial and 
for each $c \in \naturalnumbers$,
  $r \notin \bigoh(n^c)$.  
We diagonalize against nondeterministic
polynomial-time oracle Turing machine $N_j$ and deterministic
polynomial-time oracle Turing machine $M_j$. 
That is, we make sure that $L_B$ is not decided according to the 
definition of LWPP by $\gapp$ function $g$
computed by $N_j$ together with $\fp$ function $f$ computed by $M_j$
(see Definition~\ref{def:lwpp}). 
 Let $val$ be the value
computed by $M_j^{B_{j-1}}(0^{n_j})$.  
Because of the condition $0 \notin \range(f)$ in the definition of 
$\lwpp$, we 
can 
assume that $val \not= 0$.
(If $val =0$ then we can right away go to stage $j+1$.)

Let
\[
  T = \{ w \in \Sigma^{n_j} \, | \, M_j^{B_{j-1}}(0^{n_j}) \text{ queries } w \}.
\]
Note that $\|T\| \le p_j(n_j)$ since the computation time of
$M_j^{B_{j-1}}(0^{n_j})$ is bounded by $p_j(n_j)$.
In the following, we will never add any string
from $T$ to the oracle.  This ensures that the value $val$ computed by
$M_j^{B_{j-1}}(0^{n_j})$ is never changed when we replace oracle $B_{j-1}$ by $B_j$. 

$(\ast)$ We choose a set $C \subseteq \Sigma^{n_j} - T$
such that%

\begin{itemize}
   \item $\|C\|  \in \{ 1, 2, \ldots , r(n_j) \}$ and  
     $\gap_{N_j^{B_{j-1}  \cup C}}(0^{n_j}) \not= val$, or 
   \item $\|C\|  = 0$   and $\gap_{N_j^{B_{j-1}  \cup C}}(0^{n_j}) \not= 0$.
\end{itemize}
Let $B_j = B_{j-1} \cup C$.

\smallskip

\noindent
{\boldmath\bf End of Stage $j$.}

\medskip

This construction guarantees that
for each $n$, $\|B^{=n}\| \le r(n)$ and
$L_B \notin \lwpp^B$.  
Thus our proof is complete if we can 
show that it is always possible to find a set $C$ satisfying 
$(\ast )$.  We state and prove that as the following 
claim and its proof.

\begin{claim} \label{claim:can-extend}
  For each $j \ge 1$, there exists a set 
  $C$ satisfying $(\ast)$.
\end{claim}
\begin{proof}[Proof of Claim~\ref{claim:can-extend}]
Suppose that in stage $j$ no set $C$ satisfying $(\ast)$ exists.
Then for every  $C \subseteq \Sigma^{n_j} - T$, the
following holds:
\begin{align}
  & \|C\|  \in \{ 1, 2, \ldots , r(n_j) \} \Longrightarrow  
     \gap_{N_j^{B_{j-1}  \cup C}}(0^{n_j}) = val, \text{ and} \label{eq:gl4}
       \\ 
  & \|C\|  = 0 \Longrightarrow  \gap_{N_j^{B_{j-1}  \cup C}}(0^{n_j})
     = 0. \label{eq:gl3}
\end{align}
Let $s'\in \integers[y_1, y_2, \ldots , y_m]$ be the polynomial encoding of
$N_{j}^{(\cdot)}(0^{n_j})$ as in Definition~\ref{def:poly-encoding}. 
W.l.o.g.~assume that $x_1, x_2, \ldots , x_N$
enumerate the strings in $\Sigma^{n_j} - T$, and
$x_{N+1}, x_{N+2}, \ldots , x_m$ enumerate the remaining strings of
length at most
$p_j(n_j)$.
By Proposition~\ref{prop:poly-encoding}, polynomial
$s'(y_1, y_2, \ldots , y_m)$ satisfies the following:
\begin{enumerate}
  \item For all $C \subseteq \Sigma^{n_j} - T$, it holds that
    \[
        s'(\charfun_{C}(x_1), \charfun_{C}(x_2), \ldots, \charfun_{C}(x_N),
      \charfun_{B_{j-1}}(x_{N+1}), \ldots ,
      \charfun_{B_{j-1}}(x_{m}) )
          = \gap_{N_{j}^{B_{j-1} \cup C }}(0^{n_j}), \text{ and}
    \]
   \item $\deg(s') \le p_j(n_j)$.
\end{enumerate}
Note that the sets $B_{j-1} $ and $\Sigma^{n_j} - T$ are disjoint.
The values $\charfun_{B_{j-1}}(x_{N+1}), \ldots ,
      \charfun_{B_{j-1}}(x_{m})$ do not depend on $C$.
We fix these values and obtain the following polynomial $s(y_1, y_2,
\ldots , y_N)$: 
\[
  s(y_1, y_2, \ldots , y_N) = s'(y_1, y_2, \ldots , y_N, \charfun_{B_{j-1}}(x_{N+1}), \ldots , \charfun_{B_{j-1}}(x_{m}) ).
\]
Hence the polynomial $s(y_1, y_2, \ldots , y_N)$ satisfies the following:
\begin{enumerate}
  \item For all $C \subseteq \Sigma^{n_j} - T$, it holds that
    \begin{equation}
     s(\charfun_{C}(x_1), \charfun_{C}(x_2), \ldots,
     \charfun_{C}(x_N)) = \gap_{N_{j}^{B_{j-1} \cup C }}(0^{n_j}),
     \text{ and}
    \end{equation}
  \item $\deg(s) \le \deg(s') \le p_j(n_j)$.
\end{enumerate}
Statements (\ref{eq:gl3}) and (\ref{eq:gl4}) imply that
\begin{itemize}
  \item $s(0, 0, \ldots, 0)  =  0$, and
  \item for all $z_1, z_2, \ldots , z_{N}\in \{ 0, 1\}$ such that
    $\sum\limits_{1 \leq i \leq N} z_{i} \in \{ 1, 2, \ldots , r(n_j) \} $, we have $s(z_1, z_2, \ldots, z_{N})  =  val$.
\end{itemize}
By Lemma~\ref{lemma:prime-divisor}, for every prime $k$ in 
\[
S = \{ \deg(s) + 1 , \deg(s) + 2,   \ldots , \min(N/2, r(n_j)) \},
\]
it holds that $k \mid val$.

To obtain a lower bound for $val$, we determine a lower bound for the number of primes in $S$.
First, note that 
at the beginning of stage $j$, we have taken $n_j$ large enough
such that
$N/2 \ge (2^{n_j} - p_j(n_j))/2 \ge p_j(n_j)^4$ and $r(n_j) \ge p_j(n_j)^4$, and thus
$\min(N/2, r(n_j)) \ge  p_j(n_j)^4$.
In light of Lemma~\ref{lemma:prime-number-theorem}, we hence obtain 
\[
\pi(  \min(N/2, r(n_j))   ) \ge \frac{p_j(n_j)^4}{\ln(p_j(n_j)^4)} \ge p_j(n_j)^3.
\]
Further, we obviously have
\[
  \pi(\deg(s)) \le \deg(s) \le p_j(n_j).
\]
Hence the number of primes in $S$ is at least $p_j(n_j)^3 - p_j(n_j)$,
which is (by the choice of $n_j$ at the beginning of stage $j$) 
greater than or equal to $p_j(n_j)^2$.
Since each prime in $S$ is greater than or
equal to $2$ and $val \not=0$, we
have that the absolute value of $val$ is at least $2^{p_j(n_j)^2}$.
Yet we also 
must have that 
the absolute value of $val$ is less than or 
equal to $2^{p_j(n_j)}$,
since the running time of $M_j^{(\cdot )}(0^{n_j})$ is bounded
by $p_j(n_j)$ regardless of the oracle. This is a contradiction.

So for each $j \geq 1$, $B_{j-1}$ can always be extended in
stage $j$ as required.
This finishes the proof of Claim~\ref{claim:can-extend}.
\end{proof}
And since 
Claim~\ref{claim:can-extend} was all that remained in our proof of 
Theorem~\ref{thm:lower-bound},
that theorem itself is now proven.
\end{proof}

Is the oracle constructed in the proof of
Theorem~\ref{thm:lower-bound} necessarily recursive?  It might not 
be, since Theorem~\ref{thm:lower-bound} put no complexity
or computability restrictions on the function $r$.  However, 
aside from that our construction is clearly effective; so if $r$ 
is a computable function, then the oracle our construction 
builds is certainly recursive.

\section{\boldmath${{\rm LWPP^+}}$}\label{s:body-tr-lwppplus}

Theorem~\ref{thm:lwpp-robust-poly} of Section~\ref{s:many} established
a robustness property of $\lwpp$, namely, that
$\indexedpolylwpp = \lwpp$.  That is, having one target value for
acceptance and having a list of target values
for acceptance yield the same class of languages, in the content of
LWPP, which, recall, is defined in terms of the values of $\gapp$
functions.  That robustness result is itself robust in the sense 
that it holds both in the real world and, it is easy to see, in 
every relativized world.

On the other hand, this equivalence for LWPP, in terms of descriptive 
richness, between one target value and a polynomial number of values,
may not hold even for quite similar 
counting-class situations.  
In particular, in this section we prove that in some 
relativized worlds, for 
the analog of LWPP defined in terms of $\sharpp$ rather than
$\gapp$ functions, having even two target values for acceptance
yields a richer class of languages than having one value.  So
it is not the case that single targets and lists of targets inherently
function identically as to descriptive richness for counting classes.

The analog of WPP defined using $\sharpp$ functions rather than
$\gapp$ functions already exists in the literature, namely it is the class
known as $\rm F_{=}P$ that was recently 
introduced by Cox and
Pay~\cite{cox-pay:tarxiv:wppplus}.
Similarly, we here define, and denote as $\lwpp^+$, the analog of LWPP
except defined using $\sharpp$ functions rather than $\gapp$
functions.  Clearly, $\lwpp^+ \subseteq \lwpp$, and it would be
natural to guess that the containment is strict, although obviously
proving that strictness is a much stronger result than proving
$\rm P\neq PSPACE$ and so seems beyond current techniques.

We now define $\lwpp^+$ and then we prove 
as Theorem~\ref{thm:lwppplus-not-robust} that, unlike $\lwpp$, there
are relativized worlds where size-two target sets yields languages
that cannot be obtained via any size-one target set.

\begin{definition}  \label{def:lwppplus}
$\lwpp^+$ is the class of all sets $A$ such that there exists a
  function $g\in\sharpp$ and a function $f \in \fp$ that maps from $0^*$ to
  $\positivenumbers$  such that for all $x\in\sigmastar$,
\begin{align*}
  x \in A & \Longrightarrow g(x) = f(0^{|x|})\\
  x \notin A & \Longrightarrow g(x) = 0.
\end{align*}
\end{definition} 

\begin{definition}
$\twolwppplus$ is the class of all sets $A$ such that there exists a function
  $g\in\sharpp$ and 
functions $f_1, f_2 \in \fp$ that map from $0^*$ to
  $\positivenumbers$
  such that for all $x\in\sigmastar$,
\begin{align*}
  x \in A & \Longrightarrow g(x) \in \{ f_1(0^{|x|}), f_2(0^{|x|})\} \\
  x \notin A & \Longrightarrow g(x) = 0.
\end{align*}
\end{definition} 

\begin{theorem} \label{thm:lwppplus-not-robust}
There exists an oracle $\mathcal{O}$
such that 
$({\lwpp^+})^{\mathcal{O}}
\subsetneq 
(\twolwppplus)^{\mathcal{O}}$.
\end{theorem}

We will soon prove this theorem, but we first 
state and prove a corollary, and will 
give some groundwork
for the theorem's proof, and also pointers to some related work.

\begin{corollary}\label{c:plus-nonplus}
There exists an oracle $\mathcal O$ 
such that $\lwpp^{\mathcal O} \not \subseteq (\lwpp^+)^{\mathcal O}$.
\end{corollary}
\begin{proof}%
It 
follows from the definitions that, for each 
oracle $\mathcal Q$,
$(\twolwppplus)^{\mathcal Q}
\subseteq \indexedpolylwpp^{\mathcal Q}$.  
Since, as mentioned above, 
Theorem~\ref{thm:lwpp-robust-poly} clearly relativizes, 
we have that for each oracle $\mathcal Q$,
$(\twolwppplus)^{\mathcal Q} \subseteq 
\lwpp^{\mathcal Q}$.  Thus, for the oracle $\mathcal{O}$
of Theorem~\ref{thm:lwppplus-not-robust}, 
$\lwpp^{\mathcal O}
\not\subseteq (\lwpp^+)^{\mathcal O}$.
\end{proof}

Let us give the definition of the complexity class
that is now 
known as $\rm UP_{\leq 2}$.  This class was 
introduced by 
Beigel~\cite{bei:c:up1} (in which the class 
was denoted $\rm U_2P$).

\begin{definition}
$\upleqtwo$ is the class of all sets $A$ such that there exists a function
  $g\in\sharpp$
  such that for all $x\in\sigmastar$,
\begin{align*}
  x \in A & \Longrightarrow g(x) \in \{1,2\} \\
  x \notin A & \Longrightarrow g(x) =0.
\end{align*}
\end{definition} 

Just as Valiant's class UP~\cite{val:j:checking} (``unambiguous NP'')
is the restriction of NP to allowing at most one accepting path,
Beigel's $\upleqtwo$ is the restriction of NP to allowing at most two
accepting paths.  Analogously to the fact that 1-to-1
complexity-theoretic one-way functions exist if and only if
$\rm P \neq UP$~\cite{gro-sel:j:complexity-measures,ko:j:operators},
it holds that 2-to-1 complexity-theoretic one-way functions exist if
and only if
$\rm P \neq
\upleqtwo$~\cite{hem-zim:tOutByJourExceptUPkStuffOnlyIsHere:balanced}.

Note that $\upleqtwo \subseteq \twolwppplus$. %
To prove Theorem~\ref{thm:lwppplus-not-robust}, we will 
prove the stronger result that there is an oracle relative to which even 
$\upleqtwo$ is not contained in 
contained in $\lwpp^+$.
\begin{theorem} \label{thm:lwppplus-not-robust-special-case}
There exists an oracle $\mathcal{O}$
such that 
$\upleqtwo^{\mathcal{O}} \not\subseteq ({\lwpp^+})^{\mathcal{O}}$. 
\end{theorem}

\begin{proof}
First, we need a test language.  
For every set $B \subseteq \sigmastar$, define $L_B$ as
\[
  L_B = \{ 0^n \, | \, \|B^{=n}\| > 0 \}.
\]
If $B$ satisfies the condition 
that for every length $n$ it holds that $\|B^{=n}\| \le 2$,
then $L_B \in 
\upleqtwo^B$.

In this proof, we are going to centrally use a slightly  unusual notion
of what is a computation path of a nondeterministic polynomial-time
oracle Turing machine on a given oracle 
and a given input.  Our notion of a computation path of such a
machine relative to some oracle on some input will give not just the
nondeterministic choices of the path, but will also contain a set of
strings that ideally should capture exactly the set of queries that
are asked and answered Yes along the path, and a set of strings that
ideally should capture exactly the set of queries that are asked and
answered No along the path.  However, this is made a bit trickier as
to formalizing it by the fact that we'll be dealing with various
hypothetical oracles in our proof, and so sometimes things will fail to 
be a computation path of a machine relative to an oracle due to 
the action of the machine on the given input under the given oracle not 
being in harmony with nondeterministic bits or sets that are specifying
the computation path.  So our basic notion of a path will not enforce
much on the two sets; rather, only when asserting that the path is 
in harmony with a machine under some oracle on a certain input will we
require that the ``natural'' things hold.

\begin{definition}
\begin{enumerate}
\item
A \emph{computation path} is a triple: a binary string, a set of
strings, and second set of strings that is disjoint from the first 
set of strings.

\item 
(This is implicit in the above part of the 
definition, but is here stated explicitly for clarity, as this 
will be tacitly but importantly drawn on in our proof.)
Two computation paths
$\rho_1$ and $\rho_2$ are \emph{equal} if and only if 
their first components are identical, and their second components
are identical, and their third components are identical.

\item\label{pro}
A computation path $\rho$ is said to be 
\emph{a computation path 
of a 
nondeterministic polynomial-time
oracle Turing machine $N$ with oracle $\calo$
on input $x$} if $\rho$ consists of a binary string
$\rho_{choices}$, a set of strings
$Q^+(\rho)$, and a set of strings 
$Q^-(\rho)$, that satisfy all of the following conditions.
\begin{enumerate}
\item 
$Q^+(\rho) \subseteq \calo$ and 
$Q^-(\rho) \subseteq \overline{\calo}$.\footnote{Note that 
this implies that 
$Q^+(\rho)$ and $Q^-(\rho)$ are disjoint.}
\item When one runs machine $N$ with oracle 
$Q^+(\rho)$ on input $x$, there is a path $r$ in that machine's 
computation tree such that all the following hold:
\begin{enumerate}

\item\label{foo} the  sequence of 
nondeterministic choices in $r$ is precisely 
$\rho_{choices}$ (i.e., the path makes precisely 
$|\rho_{choices}|$ (binary) choices and 
$\rho_{choices}$ lists those choices in the order they are made);
and 

\item\label{bar} the set of strings queried along path $r$ is precisely 
$Q^+(\rho) \cup Q^-(\rho)$.

\end{enumerate}
\end{enumerate}

\item If in the above we after  
\ref{foo} and \ref{bar}
add new third condition that
states ``the path $r$ accepts,'' then that is our definition of
computation path $\rho$ being 
\emph{an accepting computation path of nondeterministic
polynomial-time oracle Turing machine $N$ 
with oracle $\calo$ on input $x$}.
\end{enumerate}
\end{definition}

Given an oracle $\mathcal{O}$, a computation
path $\rho$ may or may not be a computation path of a
nondeterministic Turing machine $N$ relative to $\mathcal{O}$ on 
input $x$.  For example, if what the computation
path specifies is inconsistent
with the machine's actions on that input and oracle, then 
it would
not be a computation path of that machine with that oracle and that input.
As other examples, if relative to the machine's action on that oracle
and that input the $\rho_{choices}$ 
of a computation
path has too few or too many 
bits to be precise yield a complete, actual 
path (in the traditional sense of the word, not the more encumbered one
we are using here) through the machine's action,
or if one of the two ``$Q$'' sets contains a 
string that is not queried by the machine when it is run on $x$ with
oracle $\calo$ making nondeterministic choices $\rho_{choices}$, 
then the computation path would
not be a computation path of that machine with that oracle and that input.

Let $( N_j, M_j, p_j)_{j \ge 1}$ be an enumeration of all triples such that
$N_j$ is a nondeterministic polynomial-time oracle Turing machine, 
$M_j$ is a deterministic polynomial-time oracle Turing machine
computing a function, and $p_j$
is a monotonically increasing polynomial such that the running time of both $N_j$ and $M_j$ is
bounded by $p_j$ regardless of the oracle.  We construct the oracle $B$ in
stages.  In stage $j$, we decide the membership in $B$ of strings
of length $n_j$ and extend the initial segment $B_{j-1}$ of $B$ 
to $B_j$.  Initially, we set $B_0 = \emptyset$.

{\boldmath\bf Stage $j$, where $j \ge 1$:} 
Let $n_j$ be large enough that:
(a)~$n_j > p_{j-1}(n_{j-1})$ (to ensure that the previous stages 
are not affected), and  
(b)~$2^{n_j} - p_j(n_j) > 6p_j(n_j) + 1$.
 
We diagonalize against nondeterministic
polynomial-time oracle Turing machine $N_j$ and deterministic
polynomial-time oracle Turing machine $M_j$. 
That is, we make sure that $L_B$ is not decided according to the 
definition of $\lwpp^+$ by $\sharpp$ function $g$
computed by $N_j$ together with $\fp$ function $f$ computed by $M_j$
(see Definition~\ref{def:lwppplus}). 
 Let $val$ be the value
computed by $M_j^{B_{j-1}}(0^{n_j})$.  
Because of the condition $0 \notin \range(f)$ in the definition of 
$\lwpp^+$, we 
can 
assume that $val \not= 0$.
(If $val =0$ then we can right away go to stage $j+1$.)

Let
\[
  T = \{ w \in \Sigma^{n_j} \, | \, M_j^{B_{j-1}}(0^{n_j}) \text{ queries } w \}.
\]
Note that $\|T\| \le p_j(n_j)$ since the computation time of
$M_j^{B_{j-1}}(0^{n_j})$ is bounded by $p_j(n_j)$.
In the following, we will never add any string
from $T$ to the oracle.  This ensures that the value $val$ computed by
$M_j^{B_{j-1}}(0^{n_j})$ is never changed when we replace oracle $B_{j-1}$ by $B_j$. 

$(\ast\ast)$ We (and we will soon prove that such a set must exist) 
choose a set $C \subseteq \Sigma^{n_j} - T$
such that%

\begin{itemize}
   \item $\|C\|  \in \{ 1, 2 \}$ and  
     $\accept_{N_j^{B_{j-1}  \cup C}}(0^{n_j}) \not= val$, or 
   \item $\|C\|  = 0$   and $\accept_{N_j^{B_{j-1}  \cup C}}(0^{n_j}) \not= 0$.
\end{itemize}
Let $B_j = B_{j-1} \cup C$.

\smallskip

\noindent
{\boldmath\bf End of Stage $j$.}

This construction guarantees that
for each $n$, $\|B^{=n}\| \le 2$ and thus
$L_B \notin ({\lwpp^+})^B$.  
Thus our proof is complete if we can 
show that it is always possible to find a set $C$ satisfying 
$(\ast\ast )$.  We now state and prove that as 
Claim~\ref{claim:can-extend-lwppplus}.

\begin{claim} \label{claim:can-extend-lwppplus}
  For each $j \ge 1$, there exists a set 
  $C$ satisfying $(\ast\ast)$.
\end{claim}
\begin{proof}[Proof of Claim~\ref{claim:can-extend-lwppplus}]
Suppose that in stage $j$ no set $C$ satisfying $(\ast\ast)$ exists.
Then for every  $C \subseteq \Sigma^{n_j} - T$, the
following holds:
\begin{align}
   \|C\|  \in \{ 1, 2 \} & \Longrightarrow  
     \accept_{N_j^{B_{j-1}  \cup C}}(0^{n_j}) = val, \text{ and} \label{eq:equals-val}
       \\ 
   \|C\|  = 0 & \Longrightarrow  \accept_{N_j^{B_{j-1}  \cup C}}(0^{n_j})
     = 0. \nonumber
\end{align}
It follows that 
(i) $N_j^{B_{j-1}}(0^{n_j})$ has no accepting computation paths, and (ii)
for each string $\alpha \in \Sigma^{n_j} - T$,
 $N_j^{B_{j-1} \cup \{ \alpha\} }(0^{n_j})$ has 
exactly $val$ distinct accepting 
computation
paths
$\rho_1(\alpha), \rho_2(\alpha), \ldots, \rho_{val}(\alpha)$.
For each $\alpha \in \Sigma^{n_j} - T$, let $A_\alpha$ be the set of
accepting computation
paths of $N_j^{B_{j-1} \cup \{ \alpha\} }(0^{n_j})$.
Then for each $\alpha \in \Sigma^{n_j} - T$, we have that
$\| A_{\alpha} \| = val$.

Since $N_j^{B_{j-1}}(0^{n_j})$ has no accepting computation
paths, we claim that 
for each $\alpha_1$ and $\alpha_2$,
$\alpha_1 \neq \alpha_2$,
in $\Sigma^{n_j} - T$, it holds that 
$A_{\alpha_1} \cap A_{\alpha_2} =
\emptyset$. Why? 
Let $\alpha_1, \alpha_2 \in \Sigma^{n_j} - T$ be distinct and let
$\rho$ be any computation path in $A_{\alpha_1}$.  By definition,
$\rho$ is an accepting computation path of $N_j^{B_{j-1} \cup \{
  \alpha_1\} }(0^{n_j})$. Suppose $\alpha_1 \notin Q^+(\rho)$.  
Note that in light of that supposition we either have that 
(a)~$\alpha_1 \in Q^-(\rho)$ and $\alpha_1$ is queried by 
$N_j^{B_{j-1} \cup \{
  \alpha_1\} }(0^{n_j})$, or 
(b)~$\alpha_1 \not\in Q^+(\rho) \cup 
Q^-(\rho)$ and $\alpha_1$ is not queried by 
$N_j^{B_{j-1} \cup \{
  \alpha_1\} }(0^{n_j})$.
In each of those two cases, however, 
it is clear that 
$\rho$ will be 
an accepting computation path of $N_j^{B_{j-1}}(0^{n_j})$, which is
a contradiction because by our assumption, $N_j^{B_{j-1}}(0^{n_j})$ has
no accepting computation paths.  Hence we have shown that $\alpha_1
\in Q^+(\rho)$.  But this implies that $\rho$ is not a computation
path of $N_j^{B_{j-1} \cup \{   \alpha_2\} }(0^{n_j})$ (since
$\alpha_1
\in Q^+(\rho)$ yet 
$\alpha_1 \notin B_{j-1} \cup \{   \alpha_2\}$), and therefore $\rho
\notin A_{\alpha_2}$.

Now we define for each $\alpha \in \Sigma^{n_j} - T$
\begin{multline} \label{eq:conflicting}
  \confl(\alpha) = \{ \beta \in \Sigma^{n_j} - T \, | \,
          \mbox{ there are at least } \left\lfloor val/3 \right\rfloor +1
          \mbox{ computation paths in } A_{\alpha}\\
           \mbox{ that are not 
          computation paths of }
          N_j^{B_{j-1} \cup \{ \alpha , \beta \} }(0^{n_j})   \}.
\end{multline}
We will show that for each $\alpha \in \Sigma^{n_j} - T$,
\[
  \|\confl(\alpha)\|  \le 3 p_j(n_j).
\]
Fix some string $\alpha\in \Sigma^{n_j} - T$.
To get a contradiction, suppose that $\|\confl(\alpha)\|  > 3 p_j(n_j)$.
Let $\beta$ be some string in $\confl(\alpha)$.
Since adding $\beta$ to $B_{j-1} \cup \{ \alpha \}$ causes
at least $\left\lfloor val/3 \right\rfloor + 1$ computation paths in 
$A_{\alpha}$ to disappear from  $N_j^{B_{j-1} \cup \{ \alpha , \beta \} }(0^{n_j})$, this 
means that
\[
  \| \{ \rho\in A_{\alpha} \, | \, \beta \in Q^-( \rho) \} \|
  \ge
  \left\lfloor val/3 \right\rfloor + 1.
\]  
Then
\[
  \sum_{\beta \in \confl(\alpha)}  \| \{ \rho\in A_{\alpha} \, | \, \beta \in Q^-( \rho) \} \|
  > 3 p_j(n_j)(\left\lfloor val/3 \right\rfloor + 1) \ge p_j(n_j) \cdot val,
\]
and thus
\[
   \sum_{\beta \in  \Sigma^{n_j} - T}  \| \{ \rho\in A_{\alpha} \, | \, \beta \in Q^-( \rho) \} \| > p_j(n_j) \cdot val.
\]
Hence it follows that
\begin{equation} \label{eq:lower-bound-queries}
  \sum_{\rho \in A_{\alpha}} \|  Q^-(\rho) \| >  p_j(n_j) \cdot val.
\end{equation}
On the other hand, because each computation path $\rho$ in $A_{\alpha}$ is of
length at 
most $p_j(n_j)$, for each 
$\rho \in A_{\alpha}$,
\[
  \|Q^-(\rho)\| \le p_j(n_j).
\]
Since $\|A_{\alpha}\| = val$, it follows that
\[
   \sum_{\rho\in A_{\alpha}} \| Q^-(\rho)\| \le p_j(n_j) \cdot val,
\]
which contradicts Eqn.~(\ref{eq:lower-bound-queries}).
Hence we have shown that
for each $\alpha \in \Sigma^{n_j} - T$,
\begin{equation} \label{eqn:upper-bound-conflicting}
  \|\confl(\alpha)\|  \le 3 p_j(n_j).
\end{equation}
Because we have chosen $n_j$ large enough to ensure that  $\| \Sigma^{n_j} - T\| > 6p_j(n_j)+1$, 
Eqn.~(\ref{eqn:upper-bound-conflicting}) 
implies that there exist distinct $\gamma_1, \gamma_2 \in\Sigma^{n_j} - T$ 
such that $\gamma_1 \notin \confl(\gamma_2)$ and
 $\gamma_2 \notin \confl(\gamma_1)$.
Now consider 
\[
  N_j^{B_{j-1} \cup \{ \gamma_1 , \gamma_2 \} }(0^{n_j}).
\]
Since $\gamma_2 \notin \confl(\gamma_1)$, there are at most  
$\left\lfloor val/3 \right\rfloor $
computation
paths in  $A_{\gamma_1}$ that are not computation paths of $N_j^{B_{j-1} \cup \{ \gamma_1 , \gamma_2 \} }(0^{n_j})$.  
Further, because $\|A_{\gamma_1}\| = val$, it follows that there are at least
$val - \left\lfloor val/3 \right\rfloor$ 
computation paths in $A_{\gamma_1}$ that
{\it are} computation paths of $N_j^{B_{j-1} \cup \{ \gamma_1 , \gamma_2 \}
}(0^{n_j})$.

The same reasoning applies when we swap the roles of $\gamma_1$ and
$\gamma_2$:
Since $\gamma_1 \notin \confl(\gamma_2)$, there are at most  $\left\lfloor val/3 \right\rfloor $
computation 
paths in  $A_{\gamma_2}$ that are not computation paths of $N_j^{B_{j-1} \cup \{ \gamma_1 , \gamma_2 \} }(0^{n_j})$.  
Further, because $\|A_{\gamma_2}\| = val$, it follows that there are at least
$val - \left\lfloor val/3 \right\rfloor$ computation
paths in $A_{\gamma_2}$ that
{\it are} computation paths of $N_j^{B_{j-1} \cup \{ \gamma_1 , \gamma_2 \}
}(0^{n_j})$.

For the reasons given earlier in the proof,
$A_{\gamma_1}$ and  $A_{\gamma_2}$ are disjoint sets.
Hence there are at least $2(val - \left\lfloor val/3 \right\rfloor)$
(distinct!) computation paths in  $A_{\gamma_1} \cup  A_{\gamma_2}$
that are computation paths of $N_j^{B_{j-1} \cup \{ \gamma_1 , \gamma_2 \}
}(0^{n_j})$.  Also, all computation paths in $A_{\gamma_1}$ and
$A_{\gamma_2}$ are accepting.  Note that for any two distinct
computation paths $\rho_1, \rho_2 \in A_{\gamma_1} \cup A_{\gamma_1}$
of $N_j^{B_{j-1} \cup \{ \gamma_1 , \gamma_2 \}}(0^{n_j})$, the
corresponding nondeterministic choices ${\rho_1}_{choices}$ and 
${\rho_2}_{choices}$ must be different.  It follows that 
$N_j^{B_{j-1} \cup \{ \gamma_1 , \gamma_2 \}}(0^{n_j})$ has at least
$2(val - \left\lfloor val/3 \right\rfloor) \ge 2val - \frac{2}{3}val =
\frac{4}{3}val$ accepting computation paths.
Since 
$val \in \positivenumbers$,
this contradicts Eqn.~(\ref{eq:equals-val}) 
if we set $C = \{ \gamma_1, \gamma_2\}$.
This completes the proof of Claim~\ref{claim:can-extend-lwppplus}.
\end{proof}
Since 
Claim~\ref{claim:can-extend-lwppplus} was all that remained in our proof of 
Theorem~\ref{thm:lwppplus-not-robust-special-case},
that theorem itself is now proven.
\end{proof}

\section{\boldmath On Multiple-Target $\ceqp$}\label{s:ceqp}
We show that the class 
$\ceqp$ (see Definition~\ref{d:ceqp-def}) is robust in the sense that 
even if one changes the number of target values of acceptance-path 
cardinality from 1
to a polynomial, the class remains 
unchanged.

\begin{definition}  \label{def:ceqp-indexed}
Let $r$ be any function mapping from $\naturalnumbers$ to
$\naturalnumbers$.
Then the class
$\indexedceqp{r}$ is the class of all sets $A$ such that there exists
a nondeterministic polynomial-time Turing machine $N$
  and a function $f \in \fp$ that maps
 to
  $\integers $  such that for each $x\in\sigmastar$,
\[
  x \in A  \Longleftrightarrow \text{ there exists } i \in \{ 1, 2,
  \ldots , r(|x|)\} \text{ such that } \accept_N(x) = f(\langle x, i \rangle).
\]
\end{definition} 

\begin{definition} \label{def:indexedpolyceqp}
\[
   \indexedpolyceqp = \bigcup_{c \in
     \positivenumbers} \indexedceqp{(n^c+c)}.
\]
\end{definition}

\begin{theorem} \label{thm:ceqp-robust-poly}
$\indexedpolyceqp = \ceqp$.
\end{theorem}

\begin{proof}
It is easy to see that
$\ceqp \subseteq \indexedpolyceqp$.

To show $\indexedpolyceqp \subseteq \ceqp$, let $A$ be a set in
$\indexedpolyceqp$ defined by nondeterministic polynomial-time Turing machine $N$, $f \in \fp$, 
and polynomial
$r(n)= n^c+c$ according to Definitions~\ref{def:ceqp-indexed} and~\ref{def:indexedpolyceqp}.

Let $h_1$ be a function such that for all $x\in \sigmastar$ and $i \in\positivenumbers$,
\[
  h_1(\langle x, i \rangle) =  f(\langle x, i \rangle) -  \accept_N(x).
\]
We have $h_1 \in \gapp$ since $f \in\fp \subseteq \gapp$, $\accept_N
\in\sharpp \subseteq \gapp$, and
$\gapp$ is closed under subtraction~\cite{fen-for-kur:j:gap}.
We define $h_2$ such that for all $x\in\sigmastar$,
\[
  h_2(x) = \prod_{1 \leq i \leq r(|x|)} h_1(\langle x, i \rangle).
\]
By Closure Property~\ref{closure4}, $h_2\in \gapp$.
Note that for all $x \in \sigmastar$,
\[
  h_2(x) = \prod_{1 \leq i \leq r(|x|)} \left( f(\langle x, i \rangle) - \accept_N(x) \right).
\]
{}From the above equality, for each $x \in \sigmastar$ we have 
that $h_2(x) = 0$
if and only if 
$(\exists i \in \{ 1, 2,
  \ldots , r(|x|)\})[\accept_N(x) = f(\langle x, i \rangle)]$.
But the 
$\indexedpolyceqp$ structures defining $A$ specify that 
for 
each $x \in \sigmastar$ we have 
$x\in A$ if and only if 
$(\exists i \in \{ 1, 2,
  \ldots , r(|x|)\})[\accept_N(x) = f(\langle x, i \rangle)]$.
So $x\in A$ if and only if the $\gapp$ function 
$h_2(x)$ equals~0, and so by Theorem~\ref{thm:ceqp-as-gap} 
it follows that $A \in \ceqp$.%
\end{proof}

\section{Conclusions and Open Questions}
In this paper, we proved that LWPP and WPP are robust
enough that they remain unchanged when their single target gap is
allowed to be expanded to polynomial-sized 
lists.  We then applied this new
robustness of LWPP to show that the PP-lowness of 
the Legitimate Deck
Problem follows from a weaker hypothesis than was previously known.
In doing so, we provided enhanced evidence that 
the Legitimate Deck
Problem is not NP-hard or even NP-Turing-hard.
We also showed: that the polynomial-target robustness of LWPP that we 
have established is optimal (i.e., cannot be extended to any 
superpolynomial number of targets) with
respect to relativizable proofs; that for the \#P-based
analogue of the ($\gapp$-based) class $\lwpp$, in some relativized 
worlds even two targets give more languages than one target;
that our robustness of LWPP holds even when one simultaneously 
expands both the acceptance target-gap set and the rejection
target-gap set to be polynomial-sized lists; that our main results 
also hold for WPP; and that $\ceqp$ is polynomial-target robust.

Regarding the Reconstruction Conjecture, we proved as a consequence 
of our results that if 
there exists a polynomial $q$ such that the 
$q$-Reconstruction Conjecture holds, then the Legitimate 
Deck Problem is both PP-low and $\ceqp$-low.  Since NP is 
widely believed not to be PP-low or $\ceqp$-low, 
this provides strengthened evidence 
that the Legitimate Deck Problem is not 
NP-hard or even NP-Turing hard (since otherwise even the 
``Poly''-Reconstruction Conjecture must fail, yet even the far stronger 
Reconstruction Conjecture is generally believed to hold).

A natural open problem is whether the Legitimate Deck Problem is
$\Sigma_k^p$-low~\cite{sch:j:low}
for some $k$, i.e., whether for some $k$ it holds
that 
$(\Sigma_k^p)^{\mbox{\scriptsize{}Legitimate Deck}} = 
\Sigma_k^p$.  
Proving that that holds---though we mention that 
this has been a known open issue 
for more than two decades (see~\cite{hem-kra:j:legitimate-deck})---%
would imply that the Legitimate Deck Problem
cannot be NP-complete (even with respect to more flexible reductions
such as Turing reductions and strong nondeterministic
reductions~\cite{lon:j:nondeterministic-reducibilities})
unless the polynomial hierarchy collapses.  
%
%
%
%
%
%

\newcommand{\etalchar}[1]{$^{#1}$}

\clearpage
\appendix

\section{Proof of Theorem~\ref{thm:polypolylwpp-in-polylwpp}}

{\bf Theorem~\ref{thm:polypolylwpp-in-polylwpp}.}  
$\indexedlwpp{(\mathit{Poly},\mathit{Poly})} = \indexedpolylwpp$.

\smallskip

\noindent
{\it Proof.} 
It is easy to see that
$ \indexedpolylwpp \subseteq  \indexedlwpp{(\mathit{Poly},\mathit{Poly})} $.

To show $\indexedlwpp{(\mathit{Poly},\mathit{Poly})}\subseteq \indexedpolylwpp$, let $B$ be a set in
$\indexedlwpp{(\mathit{Poly},\mathit{Poly})}$ defined by $g \in
\gapp$, $f_A, f_R \in \fp$, 
and polynomials
$r_A(n)= n^c+c$ and $r_R(n)= n^c+c$ according to Definitions~\ref{def:rArR-LWPP}
 and~\ref{def:polypolylwpp}.

Let $h$ be a function such that for all $x\in \sigmastar$ and $j \in\positivenumbers$,
\[
  h(\langle x, j \rangle) =  f_R(\langle 0^{|x|}, j \rangle) -  g(x).
\]
We have $h \in \gapp$ since $f_R \in\fp \subseteq \gapp$, $g \in\gapp$, and
as noted earlier $\gapp$ is closed under subtraction~\cite{fen-for-kur:j:gap}.
We define $\widehat{g}$ such that for all $x\in\sigmastar$,
\[
  \widehat{g}(x) = \prod_{1 \leq j \leq r_R(|x|)} h(\langle x, j \rangle).
\]
By Closure Property~\ref{closure4}, $\widehat{g}\in \gapp$.
Note that for all $x \in \sigmastar$,
\begin{equation} \label{eq:def:hat-g}
  \widehat{g}(x) = \prod_{1 \leq j \leq r_R(|x|)} \left( f_R(\langle 0^{|x|}, j \rangle) - g(x) \right).
\end{equation}
Let $\widehat{f}$ be a function such that for all
$\ell\in\naturalnumbers$ and  $i\in\positivenumbers$,
\begin{equation} \label{eq:def:hat-f}
  \widehat{f}(\langle 0^\ell, i\rangle) = \prod_{1 \leq j \leq r_R(\ell)} \left( f_R(\langle 0^{\ell}, j \rangle) -  f_A(\langle 0^{\ell}, i \rangle) \right).
\end{equation}
It is easy to see that $\widehat{f} \in \fp$.
It follows from Eqns.~(\ref{eq:def:hat-g}) and (\ref{eq:def:hat-f}) that for every $x\in\sigmastar$,
the following are true.
\begin{enumerate}
\item If there exists $j\in \{ 1, 2, \ldots , r_R(|x|)\}$ such that
  $g(x) = f_R(\langle 0^{|x|}, j\rangle)$ then $\widehat{g}(x) = 0$.
\item For each $i \in \{ 1, 2, \ldots , r_A(|x|) \}$, it holds that
$g(x) = f_A(\langle 0^{|x|}, i \rangle )$ implies that
$\widehat{g}(x) = \widehat{f}(\langle 0^{|x|}, i \rangle )$.
\end{enumerate} 
Hence for each $x\in\sigmastar$,
\begin{align*}
  x \notin B & \Longrightarrow \widehat{g}(x) = 0,\\
  x \in B & \Longrightarrow \text{ there exists } i \in \{ 1, 2,
  \ldots , r_A(|x|)\} \text{ such that } \widehat{g}(x) = \widehat{f}(\langle 0^{|x|}, i \rangle).
\end{align*}
Note also (see Eqn.~(\ref{eq:def:hat-f})) that for all $\ell\in\naturalnumbers$ and $i \in \{ 1, 2, \ldots
, r_A(\ell)\}$, it holds that 
$\widehat{f}(\langle 0^\ell, i\rangle) \not= 0$ because for each $\ell
\in\naturalnumbers$, the sets $A_\ell$ and $R_\ell$ in
Definition~\ref{def:rArR-LWPP} are disjoint.  
By Definition~\ref{def:indexedpolylwpp}, it thus follows that $B \in \indexedpolylwpp$.
\qed

\end{document}